\documentclass[12pt,a4paper]{article}        

\usepackage{amsthm}
\usepackage{amsmath}
\usepackage{graphicx}
\usepackage{amssymb}
\usepackage{caption}
\usepackage{fullpage}

\newtheorem{thm}{Theorem}[section]
\newtheorem{lma}[thm]{Lemma}
\newtheorem{cor}[thm]{Corollary}
\newtheorem{conj}{Conjecture}

\newtheorem{prop}[thm]{Proposition}

\newtheorem{lblclaim}{Claim}
\newtheorem*{adef}{Definition}

\DeclareMathOperator{\aut}{aut}
\DeclareMathOperator{\Clique}{\#Clique}

\DeclareMathOperator{\ColSubInd}{\#ColSubInd}
\DeclareMathOperator{\SubInd}{\#SubInd}
\DeclareMathOperator{\StrEmb}{\#StrEmb}
\DeclareMathOperator{\ColStrEmb}{\#ColStrEmb}

\begin{document}

\newcommand{\leqfptT}{ $\leq^{\textup{fpt}}_{\textup{T}}$ }
\newcommand{\leqfptP}{ $\leq^{\textup{fpt}}_{\textup{pars}}$ }
\newcommand{\paramcount}[1]{\textup{\textbf{p-\#}}\textsc{#1}}
\newcommand{\paramdec}[1]{\textup{\textbf{p-}}\textsc{#1}}
\newcommand{\genproblong}{Induced Subgraph With Property}
\newcommand{\genprob}{ISWP}
\newcommand{\genprobcollong}{Multicolour Induced Subgraph with Property}
\newcommand{\genprobcol}{MISWP}

\title{The challenges of unbounded treewidth in parameterised subgraph counting problems \thanks{Research supported by EPSRC grant ``Computational Counting''}}
\author{Kitty Meeks \\
\small{School of Mathematical Sciences, Queen Mary University of London \thanks{Present address: School of Mathematics and Statistics, University of Glasgow, 15 University Gardens, Glasgow G12 8QW, UK; \texttt{\small{kitty.meeks@glasgow.ac.uk}}}} \\ 
}
\date{}
\maketitle

\begin{abstract}
Parameterised subgraph counting problems are the most thoroughly studied topic in the theory of parameterised counting, and there has been significant recent progress in this area.  Many of the existing tractability results for parameterised problems which involve finding or counting subgraphs with particular properties rely on bounding the treewidth of these subgraphs in some sense; here, we prove a number of hardness results for the situation in which this bounded treewidth condition does \emph{not} hold, resulting in dichotomies for some special cases of the general subgraph counting problem.  The paper also gives a thorough survey of known results on this subject and the methods used, as well as discussing the relationships both between multicolour and uncoloured versions of subgraph counting problems, and between exact counting, approximate counting and the corresponding decision problems.
\end{abstract}

\section{Introduction}

The field of parameterised counting complexity, first introduced by Flum and Grohe in \cite{flum04}, has received considerable attention from the theoretical computer science community in recent years; by far the most thoroughly studied family of parameterised counting problems are so-called subgraph counting problems, informally those problems which can be phrased as follows:
\\

\hangindent=0.6cm
\textit{Input:} An $n$-vertex graph $G = (V,E)$, and $k \in \mathbb{N}$. \\
\textit{Parameter:} $k$. \\
\textit{Question:} How many (labelled) $k$-vertex subsets of $V$ induce graphs with a given property?
\\

It should be noted that, although this description refers to induced subgraphs, problems in which the goal is to count copies of a given graph $H$ that are not necessarily induced can just as well be formulated in this manner: it suffices to define the property we are looking for to be ``contains $H$ as a subgraph.''

The aims of this paper are two-fold.  Firstly, we give a thorough survey of existing results relating to these problems, the methods used, and the relationships between the complexities of different problem-variants.  Secondly, we extend a method based on treewidth (previously used in the literature to demonstrate the intractability of various problems involving homomorphisms and coloured subgraphs) to give some general hardness results for exact and approximate subgraph counting problems, and also the related decision versions of these problems.

A formal model for subgraph counting problems was introduced in \cite{connected}, and this is the model considered here.  The full definition of this model, and its multicolour variant (in which the input graph is equipped with a vertex-colouring and we wish to count only those subgraphs in which each vertex has a distinct colour, as defined in \cite{bddlayers}) is given in Section \ref{model}, along with a discussion of the relationship between this model and other specific problems.

Many problems falling within the scope of this model have previously been studied in the literature, and there are a number of results concerning the intractability of exact counting for specific problems (a detailed discussion of existing results is given in Section \ref{existing}). Some general criteria which are sufficient to imply the intractability of a problem of this kind are also known (see, for example, \cite{bddlayers}), and very recently a dichotomy result for an important subcase of the model was proved by Curticapean and Marx \cite{radu14}, in which nearly all problems are shown to be intractable from the point of view of parameterised complexity.  The only tractable problems in this particular case are those for which, in order to determine whether a given $k$-vertex subset induces a graph with the desired property, it suffices to examine the edges incident with at most $c$ vertices in this subset, where $c$ is some fixed constant that does not depend on $k$; it is clear that problems of this kind can be solved in polynomial time (moreover, the simple algorithm given in \cite{radu14} can easily be extended to the more general family of problems considered in this paper).

The many intractability results for solving these parameterised counting problems exactly motivate the study of approximate counting in this setting.  The first approximation algorithm for a parameterised counting problem was given by Arvind and Raman in \cite{arvind02}, and their method has recently been extended to solve a larger family of subgraph counting problems \cite{connected}.  Both of these positive approximation results apply to problems in which the minimal graphs that satisfy the property in question have bounded treewidth; a number of positive results for the decidability of subgraph problems also exploit treewidth in a similar manner.  

In Section \ref{new} we derive a number of new intractability results which apply in the case that minimal graphs which satisfy the property under consideration do not have treewidth bounded by any constant.  Using a construction that exploits Robertson and Seymour's Excluded Grid Theorem \cite{robertson86},\footnote{This result has previously been used along the same lines in different settings by a number of authors; a discussion of the relationship of our methods to previous work is given in Section \ref{excluded-grid}.} we show that, in this particular setting:
\begin{enumerate}
\item exact counting is intractable,\footnote{This generalises a result of Curticapean and Marx \cite{radu14}, published since the writing of an initial version of this paper, which proves a special case of this result.}
\item decision is intractable in the multicolour setting, and
\item again in the multicolour setting, approximate counting is intractable.
\end{enumerate}
The last two intractability results, when combined with the existing algorithms, give rise to dichotomy results for restricted versions of the problems.

The remainder of the paper is organised as follows.  In Section \ref{prelim} we provide the necessary technical background, before giving a thorough analysis of existing results and methods in Section \ref{landscape}, as well as discussing the relationships between the various problem variants that have previously been studied.  Section \ref{new} then contains the proof of our new complexity results, and finally in Section \ref{future} we review the limitations of the existing methods and suggest directions for future research.\footnote{For a stand-alone survey, the reader may omit Section \ref{new}.}

\section{Preliminaries}
\label{prelim}

In this section, we cover the key background material which will be required throughout the rest of the paper.  We begin in Section \ref{notation} by summarising the main definitions used, and then review the key concepts we need from the theory of parameterised counting complexity in Section \ref{param-cplxty}.  In Section \ref{model} we then give the formal definition of the general model we use for subgraph counting problems, and also discuss some variants and special cases of this general problem and their relationship to specific problems studied elsewhere.

\subsection{Notation and Definitions}
\label{notation}

In this section we introduce the key terms that will be used in the remainder of the paper.

\subsubsection*{Graphs, isomorphisms, subgraphs and minors}

Throughout, all graphs are assumed to be simple, that is they do not have multiple edges or self-loops.  The order of a graph $G$, written $|G|$, is the number of vertices in $G$.  Given a graph $G = (V,E)$, and a subset $U \subseteq V$, we write $G[U]$ for the subgraph of $G$ induced by the vertices of $U$.  For any $k \in \mathbb{N}$, we write $[k]$ as shorthand for $\{1,\ldots,k\}$, and denote by $S_k$ the set of all permutations on $[k]$, that is, injective functions from $[k]$ to $[k]$.  We write $V^{(k)}$ for the set of all subsets of $V$ of size exactly $k$, and $V^{\underline{k}}$ for the set of $k$-tuples $(v_1,\ldots,v_k) \in V^k$ such that $v_1,\ldots,v_k$ are all distinct. 

Two graphs $G$ and $H$ are \emph{isomorphic}, written $G \cong H$, if there exists a bijection $\theta: V(G) \rightarrow V(H)$ so that, for all $u,v \in V(G)$, we have $\theta(u)\theta(v) \in E(H)$ if and only if $uv \in E(G)$; $\theta$ is said to be an \emph{isomorphism} from $G$ to $H$.  An \emph{automorphism} on $G$ is an isomorphism from $G$ to itself.  We write $\aut(G)$ for the number of automorphisms of $G$.  A \emph{subgraph} of $G$ is a graph $H$ which is isomorphic to a graph obtained from $G$ by possibly deleting vertices and/or edges; a \emph{spanning} subgraph is obtained by deleting edges only.  $H$ is a \emph{proper subgraph} of $G$ if $H$ is a subgraph of $G$ but is not isomorphic to $G$.

The notion of a subgraph is generalised by the definition of a \emph{graph minor}.  The graph $H$ is said to be a \emph{minor} of $G$ if there is a mapping $m: V(H) \rightarrow \mathcal{P}(V(G))$, mapping vertices of $H$ to sets of vertices from $G$, satisfying:
\begin{enumerate}
\item for all $u \in V(H)$, $G[m(u)]$ is connected,
\item for $u,v \in V(H)$ with $u \neq v$, $m(u) \cap m(v) = \emptyset$, and
\item for all $uv \in E(H)$, there exist $u' \in m(u)$ and $v' \in m(v)$ such that $u'v' \in E(G)$.
\end{enumerate}
Such a mapping $m$ is known as a \emph{minor map} from $H$ to $G$.

\subsubsection*{Coloured and labelled graphs}

If $G$ is coloured by some colouring $\omega: V \rightarrow [k]$ (not necessarily a proper colouring), we say that a subset $U \subseteq V$ is \emph{colourful} (under $\omega$) if, for every $i \in [k]$, there exists exactly one vertex $u \in U$ such that $\omega(u) = i$; note that this can only be achieved if $U \in V^{(k)}$.

We will be considering labelled graphs, where a labelled graph is a pair $(H, \pi)$ such that $H$ is a graph and $\pi : [|V(H)|] \rightarrow V(H)$ is a bijection.  We denote by $\mathcal{L}(k)$ the set of all labelled subgraphs on $k$ vertices.  

For labelled graphs $(H,\pi),(H',\pi') \in \mathcal{L}(k)$, we say that $(H,\pi)$ is a \emph{labelled subgraph} of $(H',\pi')$ if and only if, for every $uv \in E(H)$, we have $u'v' \in E(H')$, where $u' = \pi'(\pi^{-1}(u))$ and $v' = \pi'(\pi^{-1}(v))$.

Given a graph $G = (V,E)$ and a $k$-tuple of vertices $(v_1,\ldots,v_k)$, we write $G[v_1,\ldots,v_k]$ for the labelled graph $(H,\pi)$ where $H = G[\{v_1,\ldots,v_k\}]$ and $\pi(i) = v_i$ for each $i \in [k]$.  For any set $\mathcal{H}$ of labelled graphs with $(H,\pi) \in \mathcal{H}$, we set $\mathcal{H}^H = \{(H',\pi') \in \mathcal{H}: H' \cong H\}$.

\subsubsection*{Treewidth}

A key concept in this paper is the \emph{treewidth} of a graph.  We say that $(T,\mathcal{D})$ is a \emph{tree decomposition} of $G$ if $T$ is a tree and $\mathcal{D} = \{\mathcal{D}(t): t \in V(T)\}$ is a collection of non-empty subsets of $V(G)$ (or \emph{bags}), indexed by the nodes of $T$, satisfying:
\begin{enumerate}
\item $V(G) = \bigcup_{t \in V(T)} \mathcal{D}(t)$,
\item for every $e=uv \in E(G)$, there exists $t \in V(T)$ such that $u,v \in \mathcal{D}(t)$,
\item for every $v \in V(G)$, if $T(v)$ is defined to be the subgraph of $T$ induced by nodes $t$ with $v \in \mathcal{D}(t)$, then $T(v)$ is connected.
\end{enumerate}
The \emph{width} of the tree decomposition $(T,\mathcal{D})$ is defined to be $\max_{t \in V(T)} |\mathcal{D}(t)| - 1$, and the \emph{treewidth} of $G$ is the minimum width over all tree decompositions of $G$.

\subsubsection*{Embeddings}

Given two graphs $G$ and $H$, a \emph{strong embedding} of $H$ in $G$ is an injective mapping $\theta: V(H) \rightarrow V(G)$ such that, for any $u,v \in V(H)$, $\theta(u)\theta(v) \in E(G)$ if and only if $uv \in E(H)$.  We denote by $\StrEmb(H,G)$ the number of strong embeddings of $H$ in $G$.  If $\mathcal{H}$ is a class of labelled graphs on $k$ vertices, we set 
\begin{align*}
\StrEmb(\mathcal{H},G) =  \qquad \qquad  & \\
|\{\theta: [k] \rightarrow V(G) \quad : \quad & \theta \text{ is injective and } \exists (H,\pi) \in \mathcal{H}  \text{ such that } \\ &  \theta(i)\theta(j) \in E(G) \iff \pi(i)\pi(j) \in E(H)\}|.
\end{align*}
If $G$ is also equipped with a $k$-colouring $\omega$, where $|V(H)| = k$, we write $\ColStrEmb(H,G,\omega)$ for the number of strong embeddings of $H$ in $G$ such that the image of $V(H)$ is colourful under $\omega$.  Similarly, we set 
\begin{align*}
\ColStrEmb(\mathcal{H},G,\omega) = \qquad \qquad & \\
 |\{\theta:[k] \rightarrow V(G) \quad : \quad & \theta \text{ is injective, } \exists (H,\pi) \in \mathcal{H} \text{ such that } \\ 
 						& \theta(i)\theta(j) \in E(G) \iff \pi(i)\pi(j) \in E(H), \\
 						& \text{and $\theta([k])$} \text{ is colourful under } \omega\}|.
\end{align*}
We can alternatively consider unlabelled embeddings of $H$ in $G$.  In this context we write $\SubInd(H,G)$ for the number of subsets $U \in V(G)^{(|H|)}$ such that $G[U] \cong H$.  Note that $\SubInd(H,G) = \StrEmb(H,G) / \aut(H)$.  If $\mathcal{H}$ is a class of labelled graphs, we set 
\begin{align*}
\SubInd(\mathcal{H},G) = |\{U \subseteq V(G) : & \quad \exists (H,\pi) \in \mathcal{H} \text{ such that } G[U] \cong H\}|. 
\end{align*} 
Once again, we can also consider the case in which $G$ is equipped with a $k$-colouring $\omega$.  In this case $\ColSubInd(H,G,\omega)$ is the number of colourful subsets $U$ such that $G[U] \cong H$, and 
\begin{align*}
\ColSubInd(\mathcal{H},G,\omega) = |\{U \subseteq V(G) \quad : \quad & \exists (H,\pi) \in \mathcal{H} \text{ such that } G[U] \cong H,  \\
                         & \text{ and $U$ is colourful under $\omega$}\}|.
\end{align*} 
Finally, we write $\Clique_k(G)$ as shorthand for $\SubInd(K_k,G)$, where $K_k$ denotes a clique on $k$ vertices.

\subsection{Parameterised Counting Complexity}
\label{param-cplxty}

In this section, we introduce key notions from parameterised counting complexity, which we will use in the rest of the paper.  Let $Sigma$ be a finite alphabet.  A parameterised counting problem is a pair $(\Pi,\kappa)$ such that $\Pi: \Sigma^* \rightarrow \mathbb{N}_0$ is a function and $\kappa: \Sigma^* \rightarrow \mathbb{N}$ is a parameterisation (a polynomial-time computable mapping).

Just as when considering the complexity of parameterised decision problems, an algorithm to solve a parameterised counting problem is considered to be efficient if, on input of size $n$ with parameter $k$, its running time is bounded by $f(k)n^{O(1)}$, where $f$ is any computable function of $k$.  Problems (whether decision or counting problems) admitting such an \emph{fpt-algorithm} belong to the class FPT and are said to be \emph{fixed parameter tractable}.

To understand the intractability of parameterised counting problems, Flum and Grohe \cite{flum04} introduce two kinds of reductions between such problems.

\begin{adef}
Let $(\Pi,\kappa)$ and $(\Pi',\kappa')$ be parameterised counting problems.
\begin{enumerate}
\item An fpt parsimonious reduction from $(\Pi,\kappa)$ to $(\Pi',\kappa')$ is an algorithm that computes, for every instance $I$ of $\Pi$, an instance $I'$ of $\Pi'$ in time $f(\kappa(I))\cdot |I|^c$ such that $\kappa'(I') \leq g(\kappa(I))$ and 
$$\Pi(I) = \Pi'(I')$$ 
(for computable functions $f,g: \mathbb{N} \rightarrow \mathbb{N}$ and a constant $c \in \mathbb{N}$).  In this case we write $(\Pi,\kappa)$ \emph{\leqfptP} $(\Pi',\kappa')$.

\item An fpt Turing reduction from $(\Pi,\kappa)$ to $(\Pi',\kappa')$ is an algorithm $A$ with an oracle to $\Pi'$ such that
\begin{enumerate}
\item $A$ computes $\Pi$,
\item $A$ is an fpt-algorithm with respect to $\kappa$, and
\item there is a computable function $g:\mathbb{N} \rightarrow \mathbb{N}$ such that for all oracle queries ``$\Pi'(y) = ?$'' posed by $A$ on input $x$ we have $\kappa'(I') \leq g(\kappa(I))$.
\end{enumerate}
In this case we write $(\Pi,\kappa)$ \emph{\leqfptT} $(\Pi',\kappa')$.
\end{enumerate}

\end{adef}

Using these notions, Flum and Grohe introduce a hierarchy of parameterised counting complexity classes, \#W[$t$], for $t \geq 1$; this is the analogue of the W-hierarchy for parameterised decision problems.  In order to define this hierarchy, we need some more notions related to satisfiability problems.  

The definition of levels of the hierarchy uses the following class of problems, where every first order formula $\psi$ (with a free relation variable of arity $s$) gives rise to a specific problem in the class.
\\

\hangindent=1cm
\paramcount{WD}$_{\psi}$ \\
\textit{Input:} A relational structure\footnote{The relational structure $\mathcal{A}$ of vocabulary $\tau$ of relations consists of the \emph{universe} $A$ together with an interpretation $R^{\mathcal{A}}$ of every relation $R$ in $\tau$.  For further terminology we refer the reader to \cite{flumgrohe}.} $\mathcal{A}$ and $k \in \mathbb{N}$. \\
\textit{Parameter:} $k$. \\
\textit{Question:} How many relations $S \subseteq A^s$ of cardinality $|S|=k$ are such that $\mathcal{A} \models \psi(S)$ (where $A$ is the universe of $\mathcal{A}$)? \\

If $\Psi$ is a class of first-order formulas, then \paramcount{WD}-$\Psi$ is the class of all problems \paramcount{WD}$_{\psi}$ where $\psi \in \Psi$.  The classes of first-order formulas $\Sigma_t$ and $\Pi_t$, for $t \geq 0$, are defined inductively.  Both $\Sigma_0$ and $\Pi_0$ denote the class of quantifier-free formulas, while, for $t \geq 1$, $\Sigma_t$ is the class of formulas
$$\exists x_1 \ldots \exists x_i \psi,$$
where $\psi \in \Pi_{t-1}$, and $\Pi_t$ is the class of formulas
$$\forall x_1 \ldots \forall x_i \psi,$$
where $\psi \in \Sigma_{t-1}$.  We are now ready to define the classes \#W[$t$], for $t \geq 1$.

\begin{adef}[\cite{flum04,flumgrohe}]
For $t \geq 1$, $\#W[t]$ is the class of all parameterised counting problems that are fpt parsimonious reducible to a problem in \paramcount{WD}-$\Pi_t$.
\end{adef} 

Unless FPT=W[1], there does not exist an algorithm running in time $f(k)n^{O(1)}$ for any problem that is hard for the class \#W[1] under either fpt parsimonious reductions or fpt Turing reductions.  In the setting of this paper, a parameterised counting problem will be considered to be intractable if it is \#W[1]-hard with respect to either form of reduction.  A useful \#W[1]-complete problem (shown to be hard in \cite{flum04}) which we will use for reductions is \paramcount{Clique}, defined formally as follows.
\\

\hangindent=1cm
\paramcount{Clique} \\
\textit{Input:} A graph $G = (V,E)$, and $k \in \mathbb{N}$. \\
\textit{Parameter:} $k$. \\
\textit{Question:} How many $k$-vertex cliques are there in $G$? \\

When considering approximation algorithms for parameterised counting problems, an ``efficient'' approximation scheme is an FPTRAS (fixed parameter tractable randomised approximation scheme), as introduced by Arvind and Raman \cite{arvind02}; this is the analogue in the parameterised setting of an FPRAS (fully polynomial randomised approximation scheme).
\begin{adef}
An FPTRAS for a parameterised counting problem $\Pi$ with parameter $k$ is a randomised approximation scheme that takes an instance $I$ of $\Pi$ (with $|I| = n$), and rational numbers $\epsilon > 0$ and $0 < \delta < 1$, and in time $f(k) \cdot g(n,1/\epsilon,\log(1/\delta))$ (where $f$ is any function, and $g$ is a polynomial in $n$, $1/\epsilon$ and $\log(1 / \delta)$) outputs a rational number $z$ such that
$$\mathbb{P}[(1-\epsilon)\Pi(I) \leq z \leq (1 + \epsilon)\Pi(I)] \geq 1 - \delta.$$
\end{adef}

Later in the paper, we will also consider the parameterised complexity of decision problems.  This will involve the parameterised complexity class W[1], the first level of the W-hierarchy (originally introduced by Downey and Fellows \cite{downey95}), which is the decision equivalent of the \#W-hierarchy.  The levels of the W-hierarchy can be defined in a similar way to those of the \#W-hierarchy, using the following problem.
\\

\hangindent=1cm
\textsc{WD}$_{\psi}$ \\
\textit{Input:} A structure $\mathcal{A}$ and $k \in \mathbb{N}$. \\
\textit{Parameter:} $k$. \\
\textit{Question:} Is there $S \subseteq A^s$ of cardinality $|S|=k$ such that $\mathcal{A} \models \psi(S)$ (where $A$ is the universe of $\mathcal{A}$)? \\

The classes W[$t$], for $t \geq 1$, are then defined as follows.

\begin{adef}[\cite{flumgrohe}]
For $t \geq 1$, $W[t]$ is the class of all parameterised problems that are fpt-reducible to \textsc{WD}-$\Pi_t$.
\end{adef} 

The definition of fpt-reductions for parameterised decision problems is analogous to that of fpt-parsimonious for counting problems given above; for the formal definition, and other notions from the theory of paramterised complexity, we refer the reader to \cite{flumgrohe}.  When proving W[1]-hardness of various problems later in the paper, we will use reductions from the problem \paramdec{Clique}, shown to be W[1]-complete in \cite{downey95}, and formally defined as follows:
\\

\hangindent=1cm
\paramdec{Clique} \\
\textit{Input:} A graph $G = (V,E)$, and $k \in \mathbb{N}$. \\
\textit{Parameter:} $k$. \\
\textit{Question:} Does $G$ contain a $k$-vertex clique? \\

We will also be considering the relationship between randomised and deterministic parameterised algorithms.  Downey, Fellows and Regan \cite{downey98} introduced a notion of randomised fpt reductions (rephrased here for consistency of terminology):
\begin{adef}
Let $(\Pi,\kappa)$ and $(\Pi',\kappa')$ be parameterised decision problems.  A randomised (fpt, many-one) reduction from $(\Pi,\kappa)$ to $(\Pi',\kappa')$ is an algorithm that computes, for every instance $I$ of $\Pi$, an instance $I'$ of $\Pi'$ in time $f(\kappa(I))\cdot |I|^c$ such that $\kappa'(I') \leq g(\kappa(I))$ and 
\begin{enumerate}
\item if $I$ is a yes-instance for $(\Pi,\kappa)$ then the probability that $I$ is a yes-instance for $(\Pi',\kappa')$ is at least $\frac{1}{h(\kappa(I))\cdot |I|^{c'}}$, and
\item if $I$ is a no-instance for $(\Pi,\kappa)$ then $I'$ is a no-instance for $(\Pi',\kappa')$
\end{enumerate}
(for computable functions $f,g,h: \mathbb{N} \rightarrow \mathbb{N}$ and constants $c,c' \in \mathbb{N}$). 
\end{adef}

We shall use this kind of reduction to analyse the relationship between algorithms for decision and approximate counting in Section \ref{relationships} below.

\subsection{The Model}
\label{model}

This paper is concerned with the formal model for subgraph counting problems introduced in \cite{connected}.  We begin this section with the definition of the general problem and a discussion of some important special cases, as well as defining the multicolour version introduced in \cite{bddlayers}.  In Section \ref{other-probs} we then give a more detailed discussion of the relationship of this model to other problems that have previously been studied in the literature.

Let $\Phi$ be a family $(\phi_1,\phi_2,\ldots)$ of functions $\phi_k: \mathcal{L}(k) \rightarrow \{0,1\}$, such that the function mapping $k \mapsto \phi_k$ is computable.  For any $k$, we write $\mathcal{H}_{\phi_k}$ for the set $\{(H,\pi) \in \mathcal{L}(k): \phi_k(H,\pi) = 1\}$, and set $\mathcal{H}_{\Phi} = \bigcup_{k \in \mathbb{N}} \mathcal{H}_{\phi_k}$.  The general problem is then defined as follows.
\\

\hangindent=1cm
\paramcount{\genproblong}($\Phi$) (\paramcount{\genprob}$(\Phi)$) \\
\textit{Input:} A graph $G = (V,E)$ and an integer $k$.\\
\textit{Parameter:} $k$. \\
\textit{Question:} What is $\StrEmb(\mathcal{H}_{\phi_k},G)$, that is, the cardinality of the set $\{(v_1,\ldots,v_k) \in V^{\underline{k}}: \phi_k(G[v_1,\ldots,v_k]) = 1 \}$? \\

Observe that this problem can equivalently regard this problem as that of counting induced labelled $k$-vertex subgraphs that belong to $\mathcal{H}_{\Phi}$.

It was argued in \cite{connected} that this problem lies in \#W[1]:
\begin{prop}
For any $\Phi$, the problem \paramcount{\genprob}($\Phi$) belongs to \#W[1].
\label{in-W}
\end{prop}

Much of this paper will be concerned with the multicolour version of the general problem, introduced in \cite{bddlayers} (and also known to belong to \#W[1]).
\\

\hangindent=1cm
\paramcount{\genprobcollong}$(\Phi)$ (\paramcount{\genprobcol}$(\Phi)$) \\
\textit{Input:} A graph $G = (V,E)$, an integer $k$ and colouring $f: V \rightarrow [k]$.\\
\textit{Parameter:} $k$. \\
\textit{Question:} What is $\ColStrEmb(\mathcal{H}_{\phi_k},G,f)$, that is, the cardinality of the set $\{(v_1,\ldots,v_k) \in V^{\underline{k}}: \phi_k(G[v_1,\ldots,v_k]) = 1$ and $\{f(v_1),\ldots,f(v_k)\} = [k] \}$? \\

Multicolour versions of specific subgraph counting problems were initially considered as a technical tool for deriving results about the uncoloured version (as in, for example, \cite{arvind02,bddlayers}); however, as we will see in Section \ref{relationships} the multicolour versions turn out to have interesting properties when considered as separate problems in their own right.

In this paper we will also consider the decision versions of both the multicolour and uncoloured variants of the problem, which are defined as follows.
\\

\hangindent=1cm
\paramdec{\genproblong}($\Phi$) (\paramdec{\genprob}$(\Phi)$) \\
\textit{Input:} A graph $G = (V,E)$, and an integer $k$.\\
\textit{Parameter:} $k$. \\
\textit{Question:} Is there any tuple $(v_1,\ldots,v_k) \in V^{\underline{k}}$ such that \\
$\phi_k(G[v_1,\ldots,v_k]) = 1$? \\
\\

\hangindent=1cm
\paramdec{\genprobcollong}($\Phi$) (\paramdec{\genprobcol}$(\Phi)$)\\
\textit{Input:} A graph $G = (V,E)$, an integer $k$ and colouring $f: V \rightarrow [k]$.\\
\textit{Parameter:} $k$. \\
\textit{Question:} Is there any tuple $(v_1,\ldots,v_k) \in V^{\underline{k}}$ such that $\phi_k(G[v_1,\ldots,v_k]) = 1$ and $\{f(v_1),\ldots,f(v_k)\} = [k]$? \\

It is straightforward to verify that both these problems belong to the parameterised complexity class W[1].

We now consider a number of special families of properties $\Phi$ which are of particular interest.

\subsubsection*{Symmetric properties}

In Section \ref{other-probs} below, we will discuss the necessity of considering labelled graphs in order to represent a number of well-studied problems from the literature in this form.  However, an interesting special case arises when we only consider ``unlabelled'' graph problems (such as \paramcount{Clique}, or \paramcount{Connected Induced Subgraph}).  

We say that the property $\Phi$ is \emph{symmetric} if the value of $\phi_k(H,\pi)$ depends only on the graph $H$ and not on the labelling $\pi$ of the vertices.  Thus, unlabelled graph problems correspond precisely to problems \paramcount{\genprob}$(\Phi)$ with $\Phi$ symmetric.  The special case of the general problem for symmetric properties was also introduced in \cite{connected}:
\\

\hangindent=1cm
\paramcount{Induced Unlabelled Subgraph With Property}($\Phi$) \\
\textit{Input:} A graph $G = (V,E)$ and $k \in \mathbb{N}$.\\
\textit{Parameter:} $k$. \\
\textit{Question:} What is $\SubInd(\mathcal{H}_{\phi_k})$, that is, the cardinality of the set $\{\{v_1,\ldots,v_k\} \in V^{(k)}: \phi_k(G[v_1,\ldots,v_k]) = 1 \}$? \\

Observe that, for any symmetric property $\Phi$, the output of \paramcount{\genprob}$(\Phi)$ is exactly $k!$ times the output of \paramcount{Induced Unlabelled Subgraph With Property}$(\Phi)$, for any graph $G$ and $k \in \mathbb{N}$.

\subsubsection*{Monotone and uniformly monotone properties}

Much of the rest of the paper will be concerned with \emph{monotone} properties.  We say that $\Phi$ is monotone if, whenever $\phi_k(H,\pi) = 1$ and $(H,\pi)$ is a labelled spanning subgraph of $(H',\pi')$, we also have $\phi_k(H',\pi') = 1$.  It follows immediately from this definition that, if $\Phi$ is a monotone property and $G$ is a subgraph of $G'$, then 
$$\StrEmb(\mathcal{H}_{\Phi},G) \leq \StrEmb(\mathcal{H}_{\Phi},G').$$
Note that any monotone property is uniquely defined by the set of minimal labelled graphs (with respect to inclusion as labelled subgraphs) which satisfy the property, for each $k$.  

For any property $\Phi$ (not necessarily monotone), we denote by $\min(\phi_k)$ the set of minimal elements satisfying $\phi_k$ for each $k \in \mathbb{N}$, and set $\min(\Phi) = \bigcup_{k \in \mathbb{N}} \min(\phi_k)$; $\min(\Phi)$ then uniquely specifies any \emph{monotone} property, but there can be many distinct pairs of non-monotone properties $\Phi \neq \Phi'$ such that $\min(\Phi) = \min(\Phi')$.

The best studied example of a monotone property is \paramcount{Sub}$(\mathcal{C})$: this corresponds to the special case in which $|\min(\phi_k)| \leq 1$ for each $k$.  The problem \paramcount{Sub}$(\mathcal{C})$ is defined for any class of graphs $\mathcal{C} = \{H_k: k \in I_{\mathcal{C}} \subseteq \mathbb{N}\}$, where each $H_k$ has $k$ vertices (so, for each $k \in \mathbb{N}$, there is at most one element of $\mathcal{C}$ which has $k$ vertices, and the indexing set $I_{\mathcal{C}} \subseteq \mathbb{N}$ contains every $k$ such that some element of $\mathcal{C}$ has exactly $k$ vertices):
\\

\hangindent=1cm
\paramcount{Sub}($\mathcal{C}$) \\
\textit{Input:} A graph $G = (V,E)$ and $H \in \mathcal{C}$.\\
\textit{Parameter:} $k = |H|$. \\
\textit{Question:} How many copies (not necessarily induced) of $H$ are there in $G$? \\

This paper will in fact focus primarily on a sub-family of monotone properties, which satisfy an additional restriction.  For any property $\Phi$, let us write $\mathcal{H}_{\phi_k}^*$ for the set of unlabelled graphs which satisfy $\phi_k$ when equipped with an appropriate labelling.  We then set $\min^*(\phi_k)$ to be the set of minimal elements of $\mathcal{H}_{\phi_k}^*$ (with respect to inclusion as unlabelled subgraphs), and write $\min^*(\Phi) = \bigcup_{k \in \mathbb{N}} \min^*(\phi_k)$.  Note that, as illustrated in Figure \ref{non-uniform}, it is possible to define a property $\Phi$ such that, for some $(H,\pi) \in \min(\phi_k)$, $H \notin \min^*(\phi_k)$.  However, all the examples of monotone properties which have previously been studied in the literature (including \paramcount{Sub}$(\mathcal{C})$ and \paramcount{Connected Induced Subgraph}) have the property that all minimal labelled subgraphs satisfying $\Phi$ are also minimal when considered as unlabelled subgraphs (that is, $\min^*(\phi_k) = \{H: \exists \pi \text{ such that } (H,\pi) \in \min(\phi_k)\}$ for each $k$): indeed, it is hard to think of a natural subgraph counting problem that does not satisfy this condition.  We call monotone properties that behave in this way, with minimal labelled and unlabelled subgraphs satisfying the property coinciding, \emph{uniformly monotone} properties.  Our dichotomy results in Section \ref{new} apply to uniformly monotone properties; the key observation we exploit here is that, for a uniformly monotone property $\Phi$, the collection of labelled graphs $\min(\Phi)$ has bounded treewidth if and only if the collection of unlabelled graphs $\min^*(\Phi)$ has bounded treewidth.

\begin{figure}[h]
\centering
\includegraphics[width = 0.8 \textwidth]{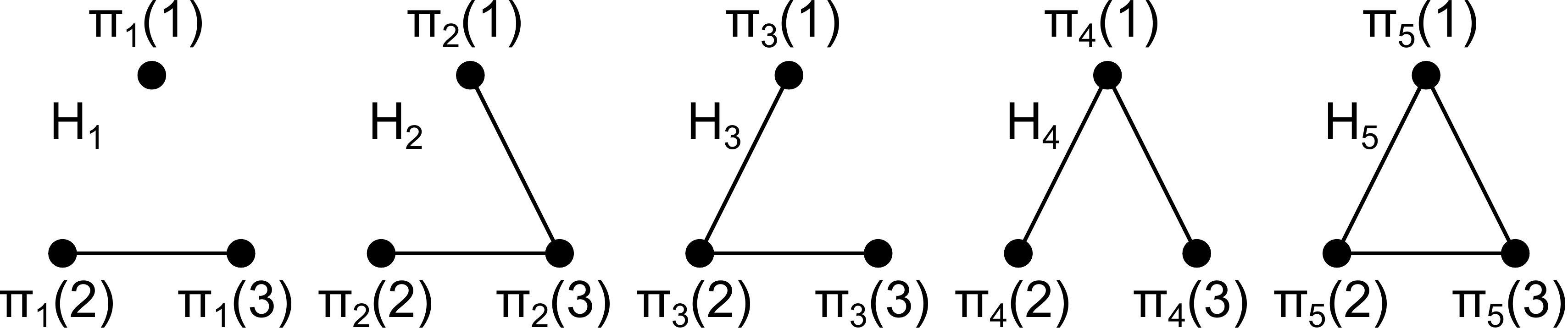}
\caption{Labelled subgraphs satisfying $\phi_3$, where $\Phi$ is not a uniformly monotone property: note that $(H_4,\pi_4)$ belongs to $\min(\phi_k)$ (it is minimal as a labelled subgraph) but not $min^*(\phi_k)$ (it is \emph{not} minimal as an \emph{unlabelled} graph).}
\label{non-uniform}
\end{figure}

\subsubsection{Relationship to other problems, and applications}
\label{other-probs}

In Section \ref{existing} below, we will see many examples of specific subgraph counting problems that have previously been studied in the literature.  All of these specific problems can be formulated in terms of the model described above (up to, possibly, division by a constant factor, easily computable in time depending only on the parameter); in this section we describe how to do this for two specific examples (as previously discussed in \cite{bddlayers}).

Our first and simplest example is \paramcount{Induced Subgraph Isomorphism}$(\mathcal{C})$, the problem of counting induced subgraphs of a graph $G$ that are isomorphic to some graph $C \in \mathcal{C}$ (where $\mathcal{C}$ is any recursively enumerable class of unlabelled graphs; to simplify the formulation, we shall assume that $\mathcal{C}$ contains at most one graph on $k$ vertices for each $k \in \mathbb{N}$).  We then set
\begin{equation*}
\phi_k(H,\pi) = \begin{cases}
					1	& \text{if $H \in \mathcal{C}$} \\
					0	& \text{otherwise.}
				\end{cases}
\end{equation*}
Then, for any $\{v_1,\ldots,v_k\} \in V^{(k)}$ such that $G[\{v_1,\ldots,v_k\}] \cong C \in \mathcal{C}$, we will have $\phi_k(G[v_{\sigma(1)},\ldots,v_{\sigma(k)}] = 1$ for every permutation $\sigma \in S_k$.  Thus, the output of \paramcount{\genprob}$(\Phi)$ on the input $(G,k)$ is equal to $k!$ times the desired value (the number of $k$-vertex induced subgraphs in $G$ that belong to $\mathcal{C}$).

Our second example is \paramcount{Sub}$(\mathcal{C})$; unlike the previous example, in this case we need to exploit the fact that the model allows us to count \emph{labelled} graphs.  We begin with a concrete special case of \paramcount{Sub}$(\mathcal{C})$, namely \paramcount{Matching} (where $\mathcal{C}$ is the set of all graphs consisting of a collection of disjoint edges).  Here, $I_{\mathcal{C}}$ is the set of all even natural numbers and, for each $k \in I_{\mathcal{C}}$, $H_k$ is the graph consisting of $k/2$ disjoint edges.  We can then set
\begin{equation*}
\phi_k(H,\pi) = \begin{cases}
					1	& \text{if $k$ is even and, for $1 \leq i \leq k/2$, $\pi(2i-1)\pi(2i) \in E(H)$} \\
					0	& \text{otherwise,}
				\end{cases}
\end{equation*}
and the output of \paramcount{\genprob}$(\Phi)$ on the input $(G,k)$ will be equal to $(k/2)! \cdot 2^{k/2}$ times the output of \paramcount{Matching} (since $(k/2)! \cdot 2^{k/2}$ is the number of automorphisms of a $k/2$-edge matching: there are $(k/2)!$ ways to map a set of $k/2$ edges to itself, and each edge can be mapped to any given edge in two different ways).

More generally, to count copies of graphs from $\mathcal{C}$ we fix, for each $k \in I_{\mathcal{C}}$, a labelling $\pi_k: [k] \rightarrow V(H_k)$, and set 
\begin{equation*}
\phi_k(H,\pi) = \begin{cases}
					1	& \text{if $k \in I_{\mathcal{C}}$ and, for every $uv \in E(H_k)$,}\\
						& \qquad \qquad \text{we have } (\pi \circ \pi_k^{-1}(u)) (\pi \circ \pi_k^{-1}(v))  \in E(H)\\
					0	& \text{otherwise.}
				\end{cases}
\end{equation*}
The output of \paramcount{\genprob}$(\Phi)$ on the input $(G,k)$ is then equal to $\aut(H_k)$ times the output of \paramcount{Sub}$(\mathcal{C})$.

Note that, if we were only able to count unlabelled subgraphs, we could not give different weight to induced unlabelled subgraphs containing different numbers of distinct copies of graphs from $\mathcal{H}$, and could only define properties corresponding to induced $k$-vertex subgraphs that contain at least $r$ copies of $H_k$.  For example, we could express the problem of counting the number of induced $k$-vertex subgraphs that contain at least one perfect matching using only unlabelled subgraphs, but to translate \paramcount{Matching} into a general framework of this kind we need to use labels.

In \cite{radu14}, the authors considered as a technical tool a related problem, \paramcount{PartitionedSub}$(\mathcal{C})$, in which $G$ and $H_k$ are equipped with vertex-colourings $f_G: V(G) \rightarrow \{1,\ldots,k\}$ and $f_H: V(H_k) \rightarrow \{1,\ldots,k\}$ respectively, and the goal is to count colourful tuples $(v_1,\ldots,v_k) \in V(G)^{\underline{k}}$ which are \emph{colour-preserving isomorphic} to $H$ with colouring $f_H$ (that is, $(f_G|_{\{v_1,\ldots,v_k\}})^{-1} \circ f_H$ defines an isomorphism from $H_k$ to $G[\{v_1,\ldots,v_k\}]$).  This problem is superficially similar to the multicolour problem \paramcount{\genprobcol}$(\Phi)$ defined above, and indeed for the special case of \paramcount{PartitionedSub}$(\mathcal{C})$ they are equivalent with respect to fpt-Turing reductions; however, this relationship does not extend to the general problem.

Results concerning this very general model can be used in any application that involves counting small substructures in a large graph.  Strategies that involve counting small substructures in large networks have previously been used for network security tools \cite{gelbord01,sekar04,staniford-chen96}, and the specific example of counting the number of triangles for the analysis of social networks has been studied extensively \cite{becchetti08,kolountzakis12,schank05,tsourakakis08}.  The counting of more general ``network motifs'' has been used to analyse biological networks \cite{milo02}, and counting versions of the \textsc{Graph Motif} problem (which relates to metabolic networks) have recently been studied in the framework of parameterised subgraph counting \cite{guillemot13,connected}.

\section{The Complexity Landscape}
\label{landscape}

We are concerned with determining the complexity of problems that can be expressed in the framework described in Section \ref{model} above, and there are three related questions we can ask about the complexities of both the uncoloured and multicolour versions (giving a total of six problem variants to consider):
\begin{enumerate}
\item Is the decision version in FPT?
\item Is there an FPTRAS for the counting version?
\item Is the exact counting problem in FPT?
\end{enumerate}
For a concrete example, consider the problem \paramcount{Connected Induced Subgraph}, studied in \cite{connected}.  In this case:
\begin{itemize}
\item exact counting is \#W[1]-complete, in both the uncoloured and multicolour versions \cite[Theorem 2.1]{connected},
\item there is an FPTRAS for the problem, in both the uncoloured and multicolour versions \cite[Corollary 3.2]{connected}, and 
\item the decision problem, in both the uncoloured and multicoloured case, is in FPT (see Corollary \ref{bddtw-mondec} below); in fact the uncoloured version can be solved in polynomial time by a component exploration method.
\end{itemize}
In Section \ref{existing}, we summarise existing results about the complexity of all six problem variants, for specific properties.  In Section \ref{methods} we then discuss the different methods that have commonly been used to address these questions in the literature, before considering in Section \ref{relationships} the relationships between the complexities of the different problem variants.

\subsection{Existing results}
\label{existing}

The study of the complexity of counting subgraphs has a long history.  Previous approaches have included the development of algorithms to solve the problem when certain restrictions are placed on either the input graph $G$ (such as its arboricity \cite{chiba85}) or the graph $H$ we wish to find (such as its flower number \cite{sundaram06}).  Exact algorithms for the problems of counting paths \cite{bjorklund09} and cycles \cite{alon97} have been studied, and recently Williams and Williams \cite{williams13} considered exact algorithms for the more general problem \paramcount{Sub}$(\mathcal{C})$.  In \cite{amini12}, Amini, Fomin and Saurabh explore the relationship between counting homomorphisms and isomorphisms to give exact algorithms for various cases of \paramcount{Sub}$(\mathcal{C})$; in Section \ref{ex-grid} below we shall discuss in more detail the relationship between our results and previous methods used to determine the complexity of homomorphism problems.

In the seminal paper on parameterised counting complexity \cite{flum04}, Flum and Grohe demonstrated that the problems \paramcount{Path} and \paramcount{Cycle} (the problems of counting $k$-vertex paths and cycles respectively, both special cases of \paramcount{Sub}$(\mathcal{C})$) are \#W[1]-complete, in contrast to the fact that the decision versions of both of these problems belong to FPT (see Theorem \ref{bddtw-subdec} below); \#W[1]-completeness of the problem of counting $k$-cliques (\paramcount{Clique}) was also proved in \cite{flum04}.  Further progress on the complexity of special cases of \paramcount{Sub}$(\mathcal{C})$ did not come for nearly a decade, until Curticapean \cite{radu13} resolved an open question from \cite{flum04}, proving the \#W[1]-completeness of \paramcount{Matching} (based on earlier joint work with Bl\"aser \cite{blaser12} on the complexity of the weighted version of this problem), another problem whose decision version belongs to FPT (again by Theorem \ref{bddtw-subdec}).  Even more recently, Curticapean and Marx gave a simpler proof of this result, showing that it remains true even if input graph $G$ is bipartite, as part of their more general dichotomy result for \paramcount{Sub}$(\mathcal{C})$: their main result shows that \paramcount{Sub}$(\mathcal{C})$ is in FPT (and is indeed solvable in polynomial time) if all graphs in $\mathcal{C}$ admit a vertex cover of size at most $c$ (where $c$ is some fixed constant, independent of the number of vertices in the graph), and otherwise this problem is \#W[1]-complete.

All of the results in the previous paragraph relate to monotone (indeed, uniformly monotone) properties.  Previously, Arvind and Raman \cite{arvind02} had demonstrated the existence of an FPTRAS for the problem \paramcount{Sub}$(\mathcal{C})$ whenever $\mathcal{C}$ is a class of graphs of bounded treewidth (implying that, among others, \paramcount{Path}, \paramcount{Cycle} and \paramcount{Matching} are efficiently approximable).  This result was recently extended by Jerrum and Meeks \cite{connected} to give an FPTRAS for \paramcount{\genprob}$(\Phi)$ whenever $\Phi$ is a monotone property and $\min(\Phi)$ is a class of graphs having bounded treewidth.  This latter result includes problems such as that of counting the number of connected $k$-vertex induced subgraphs, \paramcount{Connected Induced Subgraph}, although solving this problem exactly was also shown in \cite{connected} to be \#W[1]-complete.

Properties that are not monotone have also been studied in the literature.  Chen, Thurley and Weyer demonstrated that \paramcount{Induced Subgraph Isomorphism}$(\mathcal{C})$ is \#W[1]-complete whenever $\mathcal{C}$ contains arbitrarily large graphs; if this condition does not hold, then the problem is clearly solvable in polynomial time.  They further demonstrated that the same criterion distinguishes between the polynomial-time solvable and W[1]-complete cases of the corresponding decision problem.  Some general \#W[1]-hardness results for \paramcount{\genprob}$(\Phi)$ in the case that $\Phi$ is not monotone were also given by Jerrum and Meeks in \cite{bddlayers}, and very recently the same authors showed that \paramcount{Even Subgraph} (the problem of counting induced $k$-vertex subgraphs with an even number of edges) is \#W[1]-complete, although the problem does admit an FPTRAS, and the decision version is polynomial-time solvable.

\subsection{Methods}
\label{methods}

In this section we discuss some of the methods that have previously been used in the literature to determine the complexity of subgraph counting problems (and the related decision problems), focusing on those that have been used to address a number of different problems.

\subsubsection{Interpolation and matrix inversion}

When giving reductions to exact counting problems, certain additional techniques become available that are not applicable when considering decision or approximate counting problems; probably the most thoroughly exploited technique of this kind is polynomial interpolation, which has been widely used to prove the \#P-completeness of counting problems in the setting of classical counting complexity (see, for example, \cite{dyer00,vadhan01,valiant79,xia07}).  In essence, the reduction method involves defining a polynomial whose coefficients are in some way related to the solution of a known computationally hard problem, and which we can evaluate at different points by making suitable calls to an oracle for the problem under consideration; this makes it possible to recover the coefficients of the polynomial in question, and thus to solve the known hard problem.  The final step in such a reduction typically involves demonstrating that the resulting system of linear equations has a unique solution, or equivalently proving that a matrix is non-singular.

Techniques along these lines have also been used to obtain a number of results about the \#W[1]-completeness of parameterised counting problems, including \paramcount{WeightedMatching} \cite{blaser12} and \paramcount{Matching} \cite{radu14}.  In other parameterised subgraph counting situations, reductions have involved demonstrating the invertability of a matrix that does not arise from polynomial interpolation, as for example in proving the hardness of \paramcount{Cycle} \cite{flum04}, \paramcount{Connected Induced Subgraph} \cite{connected} and \paramcount{Even Subgraph} \cite{even}.  The strategies used to demonstrate that the resulting matrix is non-singular vary greatly in style and difficulty, and have drawn on diverse areas of mathematics including non-commutative algebra \cite{radu14} and lattice theory \cite{connected,even}.

\subsubsection{Ramsey Theory}

Ideas from Ramsey Theory have played a starring role in many results related to the parameterised complexity of counting subgraphs; here we discuss their use in proving \#W[1]-completeness, but we will see in Section \ref{sampling} below that these ideas have also been used to give positive results for decision and approximate counting.

Many of the hardness results exploit the simplest form of Ramsey's theorem; in these applications it is typically only required that the Ramsey number $R(k)$ is a computable function of $k$, and so the following bound, which follows immediately from a result of Erd\H{o}s and Szekeres \cite{erdos-szekeres}, suffices:
\begin{thm}
Let $k \in \mathbb{N}$.  Then there exists $R(k) < 2^{2k}$ such that any graph on $n \geq R(k)$ vertices contains either a clique or independent set on $k$ vertices.
\label{ramsey}
\end{thm}

The strategy in most of the reductions using Ramsey theory is to show that graphs satisfying the property $\Phi$, which we can presumably count using our oracle, must (for a large enough value of the parameter) contain arbitrarily large copies of substructures from some class (such as cliques) that we know is hard to count; with the construction of appropriate gadgets, the oracle can then be used to count (for example) cliques in the original graph.  This strategy was first exploited in the parameterised counting setting by Chen, Thurley and Weyer \cite{chen08} to address the problem \paramcount{Induced Subgraph Isomorphism}$(\mathcal{C})$.  For more complicated properties, which are satisfied by more than one $k$-vertex graph for any $k$ (as, for example, in \cite{bddlayers}), much of the work in such reductions involves constructing gadgets that ensure there is no ``confusion'' between different subgraphs that satisfy our property, so that there is indeed a correspondence between subgraphs satisfying $\Phi$ (perhaps also with certain additional properties) and the structures we wish to count in the original graph; this seems to be more easily achievable when the properties in question are in some sense ``sparse'' (as in \cite{chen08,bddlayers}).  

Such strategies sometimes give a more natural reduction to the multicolour version of the problem, which implies the hardness of the uncoloured version in the exact counting setting (as explained in Section \ref{relationships} below); reductions relying on the multicolour version as an intermediate step do not then transfer to the settings of decision or approximate counting.

Ramsey-based methods have been applied with particular success to subgraph counting problems involving hereditary classes of graphs (that is, classes that are closed under taking induced subgraphs).  Khot and Raman \cite{khot02} investigated the parameterised complexity of the decision problem \paramdec{\genprob}$(\Phi_{\mathcal{C}})$, where $\mathcal{C}$ is any hereditary graph class and $\phi_k(H,\pi) = 1$ if and only if $H \in \mathcal{C}$ with $|H| = k$.  They proved the following result:
\begin{thm}[\cite{khot02}]
If the hereditary class $\Pi$ contains all empty graphs and all complete graphs, or excludes some empty graphs and some complete graphs, then \paramdec{\genprob}$(\Phi_{\Pi})$ is in FPT; otherwise, the problem is W[1]-complete.
\label{hereditary}
\end{thm}
Both parts of this result were proved using ideas from Ramsey theory.  Very recently, Curticapean and Marx \cite{radu14} also used Ramsey-theoretic ideas to prove their dichotomy for \paramcount{Sub}$(\mathcal{C})$ in the special case that $\mathcal{C}$ is a hereditary class, proving in the process the Ramsey-style result that any graph containing a sufficiently large matching must contain either a clique on $k$ vertices, a complete bipartite graph on $k+k$ vertices, or an induced matching with $k$ edges.

In some cases, an extension of the basic Ramsey theorem is required.  Indeed, the argument in \cite{bddlayers} makes use of a straightforward corollary of Theorem \ref{ramsey} above, which gives a lower bound on the number of $k$-vertex subsets in a large graph which must induce either a clique or an independent set.
\begin{cor}[\cite{bddlayers}]
Let $G = (V,E)$ be an $n$-vertex graph, where $n \geq 2^{2k}$.  Then the number of $k$-vertex subsets $U \subset V$ such that $U$ induces either a clique or independent set in $G$ is at least
$$\frac{(2^{2k} - k)!}{(2^{2k})!}\frac{n!}{(n-k)!}.$$
\label{ramsey-cor}
\end{cor}

\subsubsection{Random sampling}
\label{sampling}

Perhaps the simplest way of devising an FPTRAS for a subgraph counting problem is to use a random sampling technique: the idea is to estimate the total number of $k$-tuples of vertices that induce graphs satisfying $\phi_k$ by picking some number of such tuples uniformly at random and counting the number of selected tuples that are ``good'' (i.e.~satisfy $\phi_k$).  However, for such a method to work, we need to know that the total number of ``good'' $k$-tuples in the input graph satisfies certain conditions.  For example, suppose that $G$ contains only one good $k$-tuple: in order to have a reasonable chance of finding this tuple (and thus distinguishing between this situation and the case in which $G$ contains no good $k$-tuple) we would have to sample $O(n^k)$ tuples, requiring more time than is allowed by the definition of the FPTRAS.

We now give a general sufficient condition on the number of good $k$-tuples in an underlying graph which will ensure that we can obtain a good approximation with a random sampling method.  The techniques are considered folklore in the classical random counting community, and so only a sketch proof is given.

\begin{lma}
Let $G=(V,E)$ be a graph on $n$ vertices and $\phi_k$ a mapping from labelled $k$-vertex graphs to $\{0,1\}$, and set $N$ to be the number of $k$-tuples of vertices $(v_1,\ldots,v_k) \in V^{\underline{k}}$ satisfying $\phi_k(G[v_1,\ldots,v_k]) = 1$.  Suppose that there exists a polynomial $q(n)$ and a computable function $g(k)$ such that either $N=0$ or $N \geq \frac{1}{g(k)q(n)} \frac{n!}{(n-k)!}$.  Then, for every $\epsilon > 0$ and $\delta \in (0,1)$ there is an explicit randomised algorithm which outputs an integer $\alpha$, such that
$$\mathbb{P}[|\alpha - N| \leq \epsilon \cdot N] \geq 1 - \delta,$$
and runs in time at most $g(k)\tilde{q}(n,\epsilon^{-1},\log(\delta^{-1}))$, where $g$ is a computable function and $\tilde{q}$ is a polynomial.
\label{exists-fptras}
\end{lma}
\begin{proof}[Proof (sketch)] 
We obtain an approximation to $N$ using a simple random sampling algorithm.  At each step, a $k$-tuple $(v_1,\ldots,v_k)$ of vertices is chosen uniformly at random among all elements of $V^{\underline{k}}$; we then determine whether $\phi_k(G[v_1,\ldots,v_k]) = 1$.  This sampling and checking step can clearly be performed in polynomial time.  To obtain a good estimate of the total number of $k$-tuples satisfying $\phi_k$, we repeat this sampling process $t$ times (for some value of $t$ to be determined), and compute the proportion $p$ of our sampled tuples which satisfy $\phi_k$.  We then output as our approximation $p \frac{n!}{(n-k)!}$.  Note that in the case that $N=0$ we are certain to output 0, as required.

The value of $t$ must be chosen to be large enough that $\mathbb{P}[|p \frac{n!}{(n-k)!} - N| \leq \epsilon \cdot N] \geq 1 - \delta$.  However, it is straightforward to verify (using, for example, a Chernoff bound) that the number of trials required is not too large (for example, $4 \log(\delta^{-1}) g(k)q(n) \epsilon^{-2}$ will do).
\end{proof}

In many cases, random sampling methods give rise to an FPTRAS for properties $\Phi$ for which the decision version is trivial (at least for sufficiently large input graphs $G$).  As an example, consider the following (symmetric) problem:
\\

\hangindent=1cm
\paramcount{Clique or Independent Set} \\
\textit{Input:} A graph $G = (V,E)$ and $k \in \mathbb{N}$.\\
\textit{Parameter:} $k$. \\
\textit{Question:} How many $k$-vertex subsets of $G$ induce either a clique or an independent set? \\

It was observed by Arvind and Raman \cite{arvind02} that the decision version of this problem is trivially in FPT: for $G$ large enough, the answer is always YES by Theorem \ref{ramsey} (and if the number of vertices in $G$ is less than $2^{2k}$, the problem can be solved by brute force search).  The approximability of the corresponding counting problem also follows easily from Lemma \ref{exists-fptras} together with Corollary \ref{ramsey-cor}:
\begin{cor}
There exists an FPTRAS for \paramcount{Clique or Independent Set}.
\label{clique+IS-app}
\end{cor}

Similarly, the proof of the positive direction of Theorem \ref{hereditary} given in \cite{khot02} invoked Ramsey theory to prove that the answer to the decision problem is trivially YES when the input graph is sufficiently large; thus, exploiting Lemma \ref{exists-fptras} again, we easily obtain the following result.
\begin{cor}
If the hereditary class $\Pi$ contains all empty graphs and all complete graphs, or excludes some empty graphs and some complete graphs, then \paramcount{\genprob}$(\Phi_{\Pi})$ admits an FPTRAS.
\end{cor}

In both of these examples, it was known that any sufficiently large graph must contain at least one good $k$-tuple; however, the formulation of Lemma \ref{exists-fptras} allows for a slightly more complicated situation, in which there is also the possibility of the input graph containing no good $k$-tuple.  This is the case when considering the problem \paramcount{Even Subgraph} in \cite{even}: Ramsey theory is exploited to show that if there is at least one induced $k$-vertex subgraph with an even number of edges then there must in fact be a large number of such subgraphs.  In this situation the decision problem is not entirely trivial, but nevertheless admits an fpt-algorithm, as it is shown that (for a sufficiently large graph $G$), if there is no $k$-vertex induced subgraph with an even number of edges, then $G$ must have a very specific and easily recognisable structure.

These methods are typically more successful in devising approximation algorithms for uncoloured problems than for the multicolour version, primarily because it seems to be much more difficult to guarantee the existence of (a large number of) specific structure(s) when an additional condition on the colouring of the structures is also imposed.  As an example, we consider the multicolour version of the problem \paramdec{Clique or Independent Set} discussed above.

\begin{prop}
\paramdec{Multicolour Clique or Independent Set} is W[1]-complete.
\label{MCclique+IS-dec}
\end{prop}
\begin{proof}
We proceed by means of a reduction from \paramdec{Multicolour Clique}, shown to be W[1]-complete in \cite{fellows09}.  Let $G$ with colouring $f$ be the input to an instance of \paramdec{Multicolour Clique}; without loss of generality we may assume that $f$ colours the vertices of $G$ with colours $[k]$.  We now construct $G'$ from $G$ by adding a new vertex $v$ adjacent to every vertex of $G$; we extend the colouring $f$, to give a colouring $f'$ of the vertices of $G'$, by setting $f(v) = k+1$.  Any colourful subset of vertices in $G'$ with colouring $f'$ must therefore contain $v$, but as $v$ is adjacent to all other vertices it cannot be part of any $k+1$-independent set.  Thus, there is a colourful clique or independent set in $G'$ (with colouring $f'$) if and only if there is a colourful clique in $G$ (with colouring $f$).
\end{proof}

\subsubsection{Color-coding}
\label{color-coding}

An extremely important technique in problems that involve finding or counting small subgraphs is that of \emph{color-coding}, first introduced by Alon, Yuster and Zwick \cite{alon95}.  This technique allows us to transform algorithms for the multicolour versions of the problems into algorithms for the uncoloured case, in the setting of decision or approximate counting.

The simplest form of this technique gives randomised algorithms.  If the vertices of $G$ are coloured independently and uniformly at random with $k$ colours, the probability that any given $k$-vertex subset is colourful with respect to this random colouring depends only on $k$.  Thus, by repeating this random colouring process sufficiently many times (which will be bounded by a function of $k$), the probability that a given $k$-vertex subset is colourful with respect to \emph{at least one} of the random colourings can be made arbitrarily high.  Thus, if the multicolour decision problem is in FPT, we can apply this known algorithm to each of the randomly coloured graphs in turn, and return the answer YES to the uncoloured decision problem if and only if we find a multicolour yes-instance with respect to at least one of the colourings.

In fact, this process can be derandomised, using \emph{$k$-perfect families of hash functions}: these are families of functions $\mathcal{F}_{n,k}$ from $[n]$ to $[k]$ such that, for every $U \in [n]^{(k)}$, there exists some $f \in \mathcal{F}_{n,k}$ such that $U$ is colourful with respect to $f$.  Bounds on the size of such families and the time required to construct them are known:
\begin{thm}[\cite{alon95}, refining constructions in \cite{schmidt90,slot84,fredman84}]
For all $n,k \in \mathbb{N}$, there is a $k$-perfect family $\mathcal{F}_{n,k}$ of hash functions from $[n]$ to $[k]$ of cardinality $2^{O(k)}\cdot\log n$, and (given $n$ and $k$) the family $\mathcal{F}_{n,k}$ can be computed in time $2^{O(k)}\cdot n \cdot \log^2 n$.
\label{perfect-families}
\end{thm}
Thus if, instead of using random colourings, we begin by constructing a $k$-perfect family of hash functions and then apply an FPT algorithm for the multicolour version with each of these colourings, we obtain a deterministic FPT algorithm for the uncoloured problem.  This observation gives the following result (implicit in \cite{alon95}):
\begin{prop}
For any family $\Phi$, if \paramdec{\genprobcol}$(\Phi)$ is in FPT then so is \paramdec{\genprob}$(\Phi)$.
\label{dec:col->uncol}
\end{prop}
This technique was used to prove the following result in \cite{alon95}; inclusion in FPT had previously been proved by Plehn and Voigt \cite{plehn90} using different techniques, but the color-coding method gives an algorithm with improved running time.
\begin{thm}
Let $\mathcal{C}$ be a class of graphs having treewidth at most $t$.  Then \paramdec{Sub}$(\mathcal{C})$ can be solved in time $2^{O(k)}|G|^{t+1}\log |G|$. 
\label{bddtw-subdec}
\end{thm}
This result covers, for example, the decision problems \paramdec{Path}, \paramdec{Cycle} and \paramdec{Matching}.  As an intermediate step in the proof of Theorem \ref{bddtw-subdec}, it is also shown that the problem of deciding whether there exists an \emph{colourful} copy of a graph from $\mathcal{C}$ is also in FPT; by applying this result for the multicolour case to each element of $\min(\Phi)$ in turn, we obtain the following more general corollary.
\begin{cor}
Let $\Phi$ be a monotone property such that $\min(\Phi)$ is a class of graphs of bounded treewidth.  Then \paramdec{\genprob}$(\Phi)$ and \paramdec{\genprobcol}$(\Phi)$ are in FPT.
\label{bddtw-mondec}
\end{cor}

When considering approximate counting rather than decision, a more careful probability analysis reveals that the same strategy of random colourings (combined with repeated calls to an appropriate approximation algorithm for the multicolour version of the problem) can be used to give an FPTRAS for the uncoloured counting problem; Arvind and Raman exploited this technique in \cite{arvind02}, implicitly using the following result.
\begin{prop}
For any family $\Phi$, if \paramcount{\genprobcol}$(\Phi)$ admits an FPTRAS, then so does \paramcount{\genprob}$(\Phi)$.
\label{app:col->uncol}
\end{prop}
Indeed, Arvind and Raman provided in \cite{arvind02} a general framework for proving the existence of an FPTRAS for a parameterised counting problem, translating the classic Karp-Luby result on approximate counting \cite{karp83} into the parameterised setting.  Specifically, they used this framework to prove the following result.
\begin{thm}[\cite{arvind02}]
Let $\mathcal{C}$ be a class of graphs of bounded treewidth.  Then there is an FPTRAS for \paramcount{Sub}$(\mathcal{C})$.
\end{thm}
This result was generalised in \cite{connected} to cover all monotone properties whose minimal elements have bounded treewidth.
\begin{thm}[\cite{connected}]
Let $\Phi$ be a monotone property such that $\min(\Phi)$ is a class of graphs of bounded treewidth.  Then there exists an FPTRAS for \paramcount{\genprob}$(\Phi)$; the same is true for \paramcount{\genprobcol}$(\Phi)$.
\label{tw-fptras}
\end{thm}

Once again, the random colouring process can be derandomised, this time using so-called \emph{$(\epsilon,k)$-balanced families} of hash functions from $[n]$ to $[k]$: such a family $\mathcal{F}_{n,k}$ has the property that there exists some constant $T$ so that, for any $U \in [n]^{(k)}$, the number of functions $f \in \mathcal{F}_{n,k}$ such that $U$ is colourful with respect to $f$ is between $(1-\epsilon)T$ and $(1+\epsilon)T$.  Constructions of such families were given in \cite{alon09,alon10}.  However, it was also shown in \cite{alon09} that the minimum size of a family of \emph{perfectly $k$-balanced} hash functions from $[n]$ to $[k]$ (that is, a $(0,k)$-balanced family) must be at least $c(k)n^{\lfloor k/2 \rfloor}$ (where $c(k)$ is a positive constant depending only on $k$), meaning that we cannot obtain an \emph{exact} fpt-algorithm for \paramcount{\genprob}$(\Phi)$ by this method even if there exists an exact fpt-algorithm for \paramcount{\genprobcol}$(\Phi)$.  In fact, it is clear that this method must fail in the exact counting case unless FPT = \#W[1], as it would give fpt-algorithms for \paramcount{Path} and \paramcount{Matching}.

\subsubsection{The Excluded Grid Theorem}
\label{ex-grid}

We have already seen that a number of positive results for decision and approximate counting rely on methods for finding copies of graphs whose treewidth is bounded.  It is natural to ask whether the use of tree decompositions to give efficient algorithms for such problems is necessary, or merely convenient: some partial answers to this question have been given by exploiting the celebrated Excluded Grid Theorem due to Robertson and Seymour \cite{robertson86}, and we use this result to explore the question more thoroughly in Section \ref{new}.
\begin{thm}[Excluded Grid Theorem \cite{robertson86}]
There is a computable function $w : N \rightarrow N$ such that the $(k \times k)$ grid is a minor of every graph of treewidth at least $w(k)$.
\label{excluded-grid}
\end{thm}

This result has been used to prove the W[1]-hardness of problems involving homomorphisms by Grohe, Schwentick and Segoufin \cite{grohe01} and by Grohe \cite{grohe07}; the latter considers the following problem, where $\mathcal{C}$ is a fixed class of recursively enumerable graphs:\footnote{In fact \cite{grohe07} addressed the more general problem of homomorphisms between arbitrary relational structures, but we restrict attention to the special case of graph homomorphisms here to avoid introducing additional notation.}
\\

\hangindent=1cm
\paramdec{Hom}$(\mathcal{C},-)$ \\
\textit{Input:} A graph $G = (V,E)$ and $H \in \mathcal{C}$.\\
\textit{Parameter:} $k = |H|$. \\
\textit{Question:} Is there a homomorphism from $H$ to $G$? \\

A homomorphism from $H$ to $G$ is a mapping $\theta: V(H) \rightarrow V(G)$ such that, for all $uv \in E(H)$, we have $\theta(u)\theta(v) \in E(G)$; the key difference from an isomorphism is that the mapping $\theta$ need not be injective in this case.  The \emph{core} of a graph $H$ is the subgraph $H'$ of $H$ such that there is a homomorphism from $H$ to $H'$, but there is no homomorphism from $H'$ to a proper subgraph of $H'$ (if more than one subgraph $H'$ of $H$ satisfies this definition, it is known that the two subgraphs must be isomorphic).  A graph $H$ has treewidth at most $t$ \emph{modulo homomorphic equivalence} if the core of $H$ has treewidth at most $t$.

In this setting, Grohe proved in \cite{grohe07} that \paramdec{Hom}$(\mathcal{C},-)$ is W[1]-hard whenever $\mathcal{C}$ is a class of recursively enumerable graphs that do not have bounded treewidth modulo homomorphic equivalence.  Since then, a number of authors have built on this result, including \cite{dalmau04,chen13,chen08,farnqvist07}.

Broadly speaking, the key idea in all of these constructions is to use a $k \times \binom{k}{2}$ grid (whose existence as a minor in one of the graphs in the class under consideration follows from Theorem \ref{excluded-grid}) to encode the existence of a clique in some other graph $G$; to achieve this, a new graph $G'$ is constructed (which will be the input to the problem to which the reduction is being made).  In $G'$, there is a vertex-class corresponding to each vertex of a $k \times \binom{k}{2}$ grid, with vertices in each class indexed by elements of $V(G) \times E(G)$; conditions are placed on edges between the vertex classes to ensure that any copy of a $k \times \binom{k}{2}$ grid in $G'$ \emph{which includes one vertex from each vertex class} must correspond to a clique in $G$.

The important point to emphasise in this description is the fact that we only want to identify grids which include one vertex from each vertex class.  Different authors have used varying strategies to ensure that the grids we identify do indeed include one vertex from each vertex class: in \cite{grohe07}, the condition that we are considering a graph whose core has unbounded treewidth means we cannot miss any of the vertex classes, whereas in \cite{grohe01} the authors considered a \emph{coloured} version of the homomorphism problem, for which unbounded treewidth alone is enough to imply W[1]-hardness.  Dalmau and Jonsson \cite{dalmau04} consider the counting version of the homomorphism problem, \paramcount{Hom}$(\mathcal{C},-)$, and show that fixed parameter tractability coincides precisely with classes $\mathcal{C}$ of bounded treewidth, using an inclusion-exclusion method (only possible in the exact counting setting) to count precisely the number of grids that use every vertex class.  Farnqvist and Jonsson studied a variant, \paramdec{LHom}$(\mathcal{C},-)$ in which restrictions are placed on the set of vertices in $G$ to which each vertex in $H$ may be mapped by the homomorphism (resulting in a problem that generalises the well-known \textsc{ListColouring} for graphs), which again allows grids that do not use every vertex class to be avoided, so the condition for W[1]-hardness is again that $\mathcal{C}$ has bounded treewidth.

Very recently, Curticapean and Marx also used these ideas to prove \#W[1]-completeness of \paramcount{Sub}$(\mathcal{C})$ when $\mathcal{C}$ is a class of graphs of unbounded treewidth; this is a special case of our result on exact counting given in Theorem \ref{exact-hard} below.  The proof of their result involves proving hardness of the problem \paramcount{PartitionedSub}$(\mathcal{C})$ (analogous to a counting version of \paramdec{LHom}$(\mathcal{C},-)$ in the discussion above) as an intermediate step; this problem is then reduced to \paramcount{Sub}$(\mathcal{C})$.

In Section \ref{new} below, we will also exploit colourings to allow us to identify those grids that include one vertex from each vertex class.  Due to a general result in Section \ref{relationships} below (Lemma \ref{uncol-col}) we can reduce the multicolour version to the uncoloured version when considering exact counting, but this is not possible for approximate counting or the decision version, so our results are only valid for the multicolour problem in these settings.

\subsection{Relationships between results}
\label{relationships}

We have already seen in Section \ref{methods} above certain relationships between the complexities of the different problem variants; in this section we derive some further implications between the complexities of the different variants of subgraph counting problems, and then summarise all the known relationships in Figure \ref{complexity-summary}.

\subsubsection*{Relationships between exact counting, approximate counting and decision}

We begin by noting the trivial fact that, if the counting version of a problem is in FPT, then it admits an FPTRAS and the corresponding decision problem is also in FPT.  Similarly, we now prove that, if a counting problem admits an FPTRAS, then there is a \emph{randomised} algorithm for the corresponding decision problem.

\begin{prop}
Let $\Phi$ be a family $(\phi_1,\phi_2,\ldots)$ of functions $\phi_k: \mathcal{L}(k) \rightarrow \{0,1\}$, such that the function mapping $k \mapsto \phi_k$ is computable, and suppose that there exists an FPTRAS for \paramcount{\genprob}$(\Phi)$.  Then, given a graph $G$ on $n$ vertices and $k \in \mathbb{N}$, there is a randomised algorithm, running in time $f(k)\cdot n^{O(1)}$, which
\begin{enumerate}
\item if $(G,k)$ is a yes-instance for \paramdec{\genprob}$(\Phi)$, returns YES with probability greater than $1/2$, and
\item if $(G,k)$ is a no-instance for \paramdec{\genprob}$(\Phi)$, returns NO.
\end{enumerate}
The same is true for the corresponding multicolour problems.
\label{random-dec}
\end{prop}
\begin{proof}
Suppose that there exists an FPTRAS for \paramcount{\genprob}$(\Phi)$.  Running this randomised approximation algorithm with $\epsilon = 1/2$ and $\delta = \frac{1}{2n+1}$ immediately gives a randomised fpt-algorithm $A$ for the decision problem with two-sided error probability strictly smaller than $\frac{1}{2n}$: if the answer to the decision problem is NO then the approximation algorithm must output ``0'' with probability at least $1 - \frac{1}{2n+1}$, whereas if the answer to the decision problem is YES then the approximation algorithm must output a value of at least $1/2$ with probability at least $1 - \frac{1}{2n+1}$.  It remains to show that this two-sided error can be reduced to a one-sided error, without increasing the probability of error too much.

We do this using standard techniques, similar to those employed to show that certain problems (so-called ``self-reducible'' problems) in the class BPP are in fact contained in RP; the idea is that before outputting the answer YES we will search for a witness to the fact that the input is indeed a yes-instance.  We begin by running our randomised decision algorithm $A$; if the output of $A$ is NO then we return NO.  If, on the other hand, the output of $A$ is YES, then we begin our search for a witness.

To do this, we consider each vertex in turn, and either delete it or add it to a list of ``necessary'' vertices.  We begin by choosing an arbitrary vertex $v$ of $G$, which we (temporarily) delete, and then run $A$ on the resulting graph.  If the output of $A$ on this new graph is YES, we permanently delete $v$ and then repeat the process (but always ensuring that we do not delete a vertex we have previously marked as ``necessary''); if the output is NO then we reinstate the vertex $v$, marking it as ``necessary'', and repeat.  This process terminates when either:
\begin{enumerate}
\item $k+1$ vertices have been marked as necessary, in which case we return NO, or
\item there are exactly $k$ vertices remaining, in which case we check (in time depending only on $k$) whether $\phi_k$ is true on any tuple formed of the remaining vertices; if so we return YES, and otherwise we return NO.
\end{enumerate}
It is clear that this witness search process is fixed parameter tractable, and moreover that it will only return YES if we do indeed have a yes-instance.  To bound the probability of returning NO if we actually have a yes-instance, observe that we call the algorithm $A$ at most $n$ times (where $n$ is the number of vertices in $G$), so by the union bound the probability that the algorithm will return the incorrect answer on at least one of these calls is at most $\frac{n}{2n+1} < \frac{1}{2}$; if $A$ returns the correct answer on all of the calls then we are sure to return YES on any yes-instance.  This gives the required algorithm in the uncoloured case.

The corresponding result for the multicolour versions of the problems follows by exactly the same argument.
\end{proof}

Subject to the assumption that W[1] is not equal to FPT under randomised parameterised reductions (as described in Section \ref{param-cplxty}), this gives the following immediate corollary.

\begin{cor}
Assume that W[1] is not equal to FPT under randomised parameterised reductions, and let $\Phi$ be a family $(\phi_1,\phi_2,\ldots)$ of functions $\phi_k: \mathcal{L}(k) \rightarrow \{0,1\}$, such that the function mapping $k \mapsto \phi_k$ is computable.  Then, if \paramdec{\genprob}$(\Phi)$ is W[1]-complete, there is no FPTRAS for \paramcount{\genprob}$(\Phi)$.  Similarly, if \paramdec{\genprobcol}$(\Phi)$ is W[1]-complete, there is no FPTRAS for \paramcount{\genprobcol}$(\Phi)$.
\label{no-decision}
\end{cor}

Recently, more refined models of parameterised random complexity have been developed by Montoya and M\"{u}ller \cite{montoya13} and Chauhan and Rao \cite{chauhan14}, in which the use of randomness is restricted.  However, given that the definition of an FPTRAS does not place any restrictions on the nature of the randomised algorithm, the existence of an FPTRAS for a problem does not have immediate implications in this new framework.  Possible future research directions to address this issue are discussed in Section \ref{future} below.

\subsubsection*{Relationships between multicolour and uncoloured problems}

The relationships between the complexities of multicolour and uncoloured versions of the problems are somewhat more complicated.  For the case of exact counting, the following result was proved in \cite{connected}, using an inclusion-exclusion method reminiscent of that used in \cite{dalmau04} (and discussed in Section \ref{ex-grid} above).

\begin{prop}
For any family $\Phi$, we have \paramcount{\genprobcol}$(\Phi)$ \leqfptT \paramcount{\genprob}$(\Phi)$.
\label{uncol-col}
\end{prop}

However, we saw in Section \ref{color-coding} above (Propositions \ref{dec:col->uncol} and \ref{app:col->uncol}) that this relationship between the multicolour and uncoloured variants for any given property $\Phi$ is reversed in the cases of approximate counting and decision.

All of these relationships between the complexities of the six different problem variants are summarised in Figure \ref{complexity-summary}: an arrow from problem $A$ to problem $B$ implies that (subject to standard assumptions in parameterised complexity theory), if problem $A$ is efficiently solvable (either it is in FPT, or it admits an FPTRAS, as appropriate), then there must also be an efficient algorithm for problem $B$.  A crossed out arrow from $A$ to $B$ indicates that there is some property $\Phi$ for which variant $A$ has an efficient algorithm but variant $B$ does not.

\begin{figure}
\centering
\captionsetup{singlelinecheck=off}
\caption[complexity implications]
{Summary of the complexity implications between the problems considered
\begin{enumerate}
\item Trivial
\item Proposition \ref{random-dec} (randomised decision algorithm)
\item Proposition \ref{uncol-col} (inclusion-exclusion)
\item Proposition \ref{dec:col->uncol} (color-coding)
\item Proposition \ref{app:col->uncol} (color-coding)
\item Counterexample: \paramdec{Clique or Independent Set} (\cite{arvind02} and Proposition \ref{MCclique+IS-dec})
\item Counterexample: \paramcount{Clique or Independent Set} (Corollary \ref{clique+IS-app} and Proposition \ref{MCclique+IS-dec}, together with Corollary \ref{no-decision})
\item Counterexample: \paramcount{Connected Induced Subgraph} (\cite{connected})
\item Counterexample: \paramcount{Sub}$(\mathcal{C})$ for any class $\mathcal{C}$ of bounded treewidth but unbounded vertex-cover number, e.g.~\paramcount{Path} (\cite{arvind02} and \cite{radu14})
\end{enumerate}}
\includegraphics[width = 0.7 \linewidth]{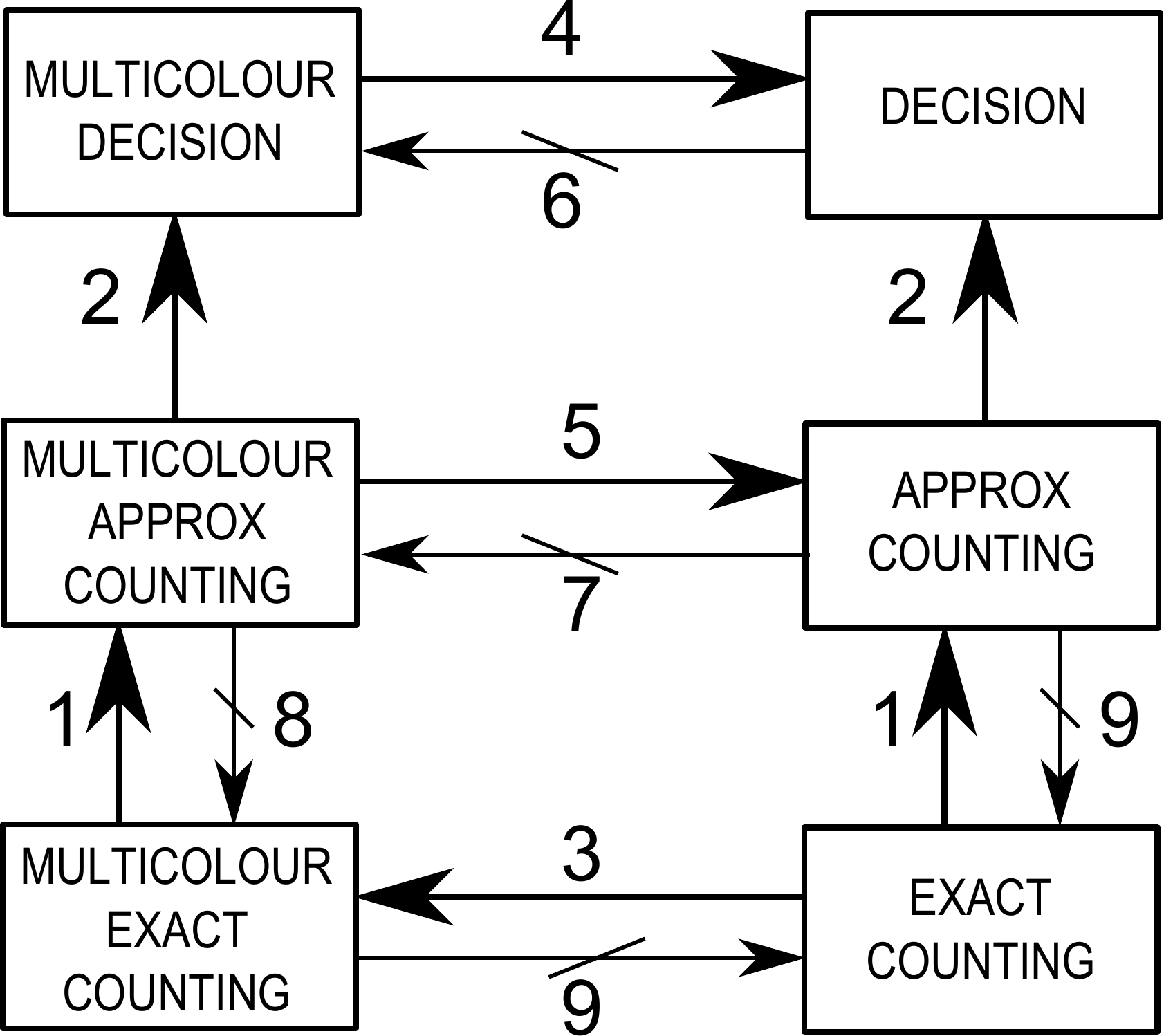}
\label{complexity-summary}
\end{figure}

\section{New results}
\label{new}

In this section, we extend methods based on the Excluded Grid Theorem, discussed in Section \ref{ex-grid} above, to give new hardness results for several of the problem variants.  Specifically, we prove that the following results hold whenever $\min^*(\Phi)$ is a class of graphs having unbounded treewidth.
\begin{enumerate}
\item Both \paramcount{\genprob}$(\Phi)$ and \paramcount{\genprobcol}$(\Phi)$ are \#W[1]-complete (Theorem \ref{exact-hard}).
\item The multicolour decision problem, \paramdec{\genprobcol}$(\Phi)$ is W[1]-complete (Theorem \ref{decision-hard}).
\item Assuming that W[1] is not equal to FPT under randomised parameterised reductions, there is no FPTRAS for \paramcount{\genprobcol}$(\Phi)$ (Corollary \ref{approx-hard}).
\end{enumerate}
Recall that inclusion of these problems in the appropriate classes was discussed in Section \ref{model} above, so it suffices for the first two results to prove the appropriate hardness results.

Our hardness results apply to any property $\Phi$ satisfying the condition that $\min^*(\Phi)$ has unbounded treewidth.  We will later restrict our attention to uniformly monotone properties, allowing us to give dichotomy results for decision and approximate counting in the multicolour setting, when the hardness results are combined with existing positive results discussed above; the dichotomies are stated explicitly in Section \ref{complexity-results} below.

All of these results are based on a common construction, described and analysed in Section \ref{grid-minors}, which exploits the Excluded Grid Theorem (Theorem \ref{excluded-grid}).  Making use of the properties of the construction established in Section \ref{grid-minors}, the proofs of the complexity results given in Section \ref{complexity-results} are reasonably straightforward.

Before embarking on the proof of these results, we consider briefly some examples of properties to which they are applicable, namely properties whose sets of minimal satisfying elements do not have bounded treewidth.  Some obvious examples come from properties that in some way depend on the treewidth of the graph: for example, if we set $\phi_k(H,\pi) = 1$ if and only if $H$ has treewidth at least $\theta(k)$, where $\theta$ is any function such that $\theta(k) \rightarrow \infty$ as $k \rightarrow \infty$, then it is clear that there is no constant that bounds the treewidth of elements of $\min(\Phi)$.  The same is true if, for example, we set $\phi_k(H,\pi) = 1$ if and only if the chromatic number of $H$ is at least $\theta(k)$ (where again $\theta(k) \rightarrow \infty$ as $k \rightarrow \infty$), since any graph with treewidth at most $t$ has chromatic number at most $t+1$ (a greedy colouring with $t+1$ colours can be obtained by considering the bags of the indexing tree in a post-order traversal, with respect to some arbitrarily chosen root).

It is also clear that our results apply when $\phi_k$ is true only if $H$ contains a complete bipartite spanning subgraph.  This category of problems includes the case in which $\phi_{k} (H) = 1$ if and only if $k$ is even and $H$ contains $K_{\frac{k}{2},\frac{k}{2}}$ as a subgraph; the parameterised complexity of this problem from the point of view of decision (in the uncoloured case) was a long-standing open problem, recently resolved by Lin \cite{lin14}.

Slightly less obvious examples of properties $\Phi$ for which $\min(\Phi)$ does not have bounded treewidth arise when $\phi_k$ is true on a $k$-vertex subgraph if and only if certain connectivity conditions are satisfied: while \paramcount{Connected Induced Subgraph}, studied in \cite{connected}, is an example of a property for which all minimal elements do have bounded treewidth (since the minimal elements satisfying $\phi_k$ are precisely all trees on $k$ vertices), if we instead demand stronger connectivity conditions this is no longer the case.  The graph obtained by subdividing every edge of the $k \times k$ grid exactly once (that is, replacing every edge of the grid with a new vertex adjacent to both endpoints of the original edge) has treewidth at least $k$, as it contains the $k \times k$ grid as a minor, and is also a minimally 2-connected graph, in the sense that it is 2-connected but deleting any edge will result in a graph that is no longer 2-connected (as it will contain a vertex of degree one).  Similarly, the graph illustrated in Figure \ref{3-connected} has treewidth at least $k$ and is 3-connected, but deleting any edge will result in a graph that is no longer 3-connected (as such a subgraph will contain at least one vertex of degree at most 2).  Both of these problems were already shown to be hard from the point of view of decision in the uncoloured case in \cite{betzler11}, in which the authors proved that the problem of deciding whether a graph contains an $r$-connected induced subgraph on $k$ vertices is W[1]-complete.  

Combining any of these properties with an additional condition on, say, the number of edges in the subgraph leads to non-monotone properties to which our results apply: a concrete example would be a property $\Phi$ such that $\phi_k(H,\pi)$ is true if and only if $H$ is 2-connected and has an even number of edges.

\begin{figure}[h]
\centering
\includegraphics[width = 0.7 \linewidth]{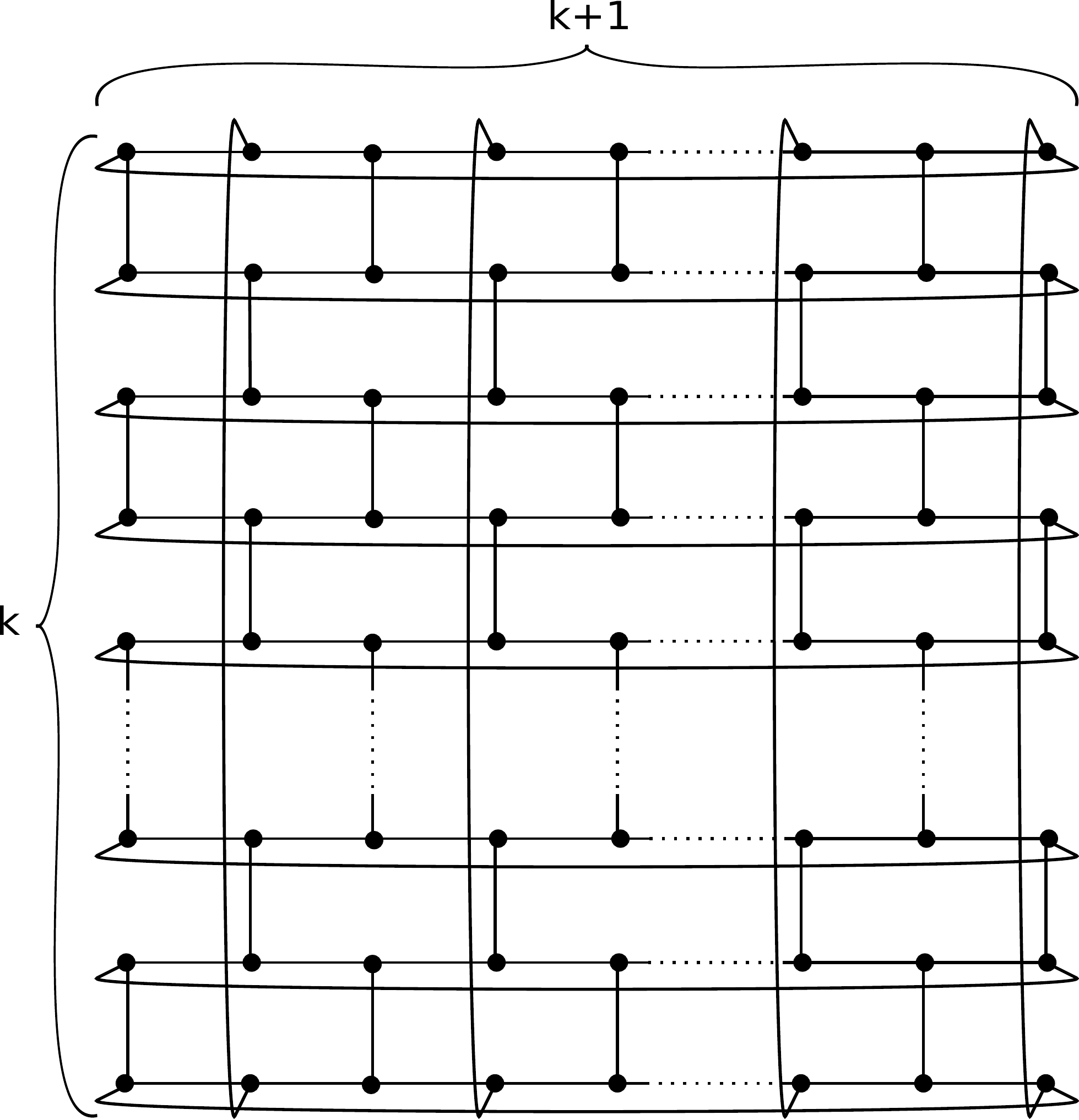}
\caption{A minimally 3-connected graph with treewidth at least $k$.}
\label{3-connected}
\end{figure}

\subsection{Graphs with grid minors}
\label{grid-minors}

In this section we present a construction based on a graph which contains a large grid minor; we define the construction in Section \ref{construction}, then analyse its key properties in Section \ref{analysis}.  This construction will form the basis of the complexity results mentioned above and proved in Section \ref{complexity-results} below.  The goal of the construction is to define, given an arbitrary graph $G$ and a graph $H$ containing a large grid as a minor, a new coloured graph $G_H$ such that the number of colourful copies of $H$ in $G_H$ and the number of $k$-cliques in $G$ are closely related.

Before presenting this construction, we make a simple observation about situations in which we will be able to apply this construction.

\begin{prop}
Let $\Phi$ be a family $(\phi_1,\phi_2,\ldots)$ of functions $\phi_k: \mathcal{L}(k) \rightarrow \{0,1\}$, such that the function mapping $k \mapsto \phi_k$ is computable.  Suppose that $\min^*(\Phi)$ is a class of graphs having unbounded treewidth.  Then, for every $j \in \mathbb{N}$ there exists $k \in \mathbb{N}$ such that some $H \in \min^*(\phi_k)$ contains the $(j \times \binom{j}{2})$ grid as a minor; moreover, the value of this $k$ is bounded by some computable function of $j$.
\label{contains-grid}
\end{prop}
\begin{proof}
By the Excluded Grid Theorem (Theorem \ref{excluded-grid}), there exists a computable function $w$ such that every graph with treewidth at least $w(\binom{j}{2})$ contains as a minor the $(\binom{j}{2} \times \binom{j}{2})$ grid, and hence the $(j \times \binom{j}{2})$ grid.  The fact that $\min^*(\Phi)$ has unbounded treewidth means that there is a computable function $g:\mathbb{N} \rightarrow \mathbb{N}$ such that, for any $w(\binom{j}{2}) \in \mathbb{N}$, there exists $k \leq g(w(\binom{j}{2}))$ such that some $H \in \min^*(\phi_k)$ has treewidth at least $w(\binom{j}{2})$ (observe that $g$ can be computed by considering $\phi_1,\phi_2,\ldots$ in turn).  Thus, $H \in \min^*(\phi_k)$ contains the $(j \times \binom{j}{2})$ grid as a minor, as required.
\end{proof}

\subsubsection{Construction}
\label{construction}

In this section we describe the construction that will be used to prove our complexity results in Section \ref{complexity-results} below.

Let $G$ be any graph, and $k \in \mathbb{N}$, and let $H$ be a graph which contains the $k \times \binom{k}{2}$ grid as a minor.  (Later, we will be interested in determining the number of $k$-cliques in $G$.)  We then define a new coloured graph $G_H$, and a colouring $f$ of its vertices with $|V(H)|$ colours.  In order to define our construction, it will be necessary to fix an arbitrary total order $\prec$ on $V(G)$.  To assist in the explanation of this construction, an example of a pair of graphs $G$ and $H$, for $k=3$, is illustrated in Figure \ref{constr-example}.

\begin{figure}
\centering
\includegraphics[width = 0.25 \linewidth]{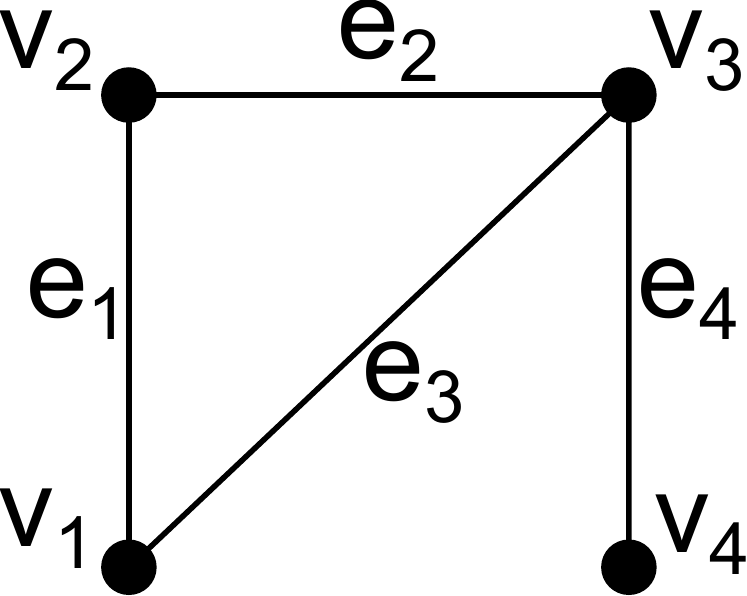} \\
\vspace{1cm}
\includegraphics[width = 0.45 \linewidth]{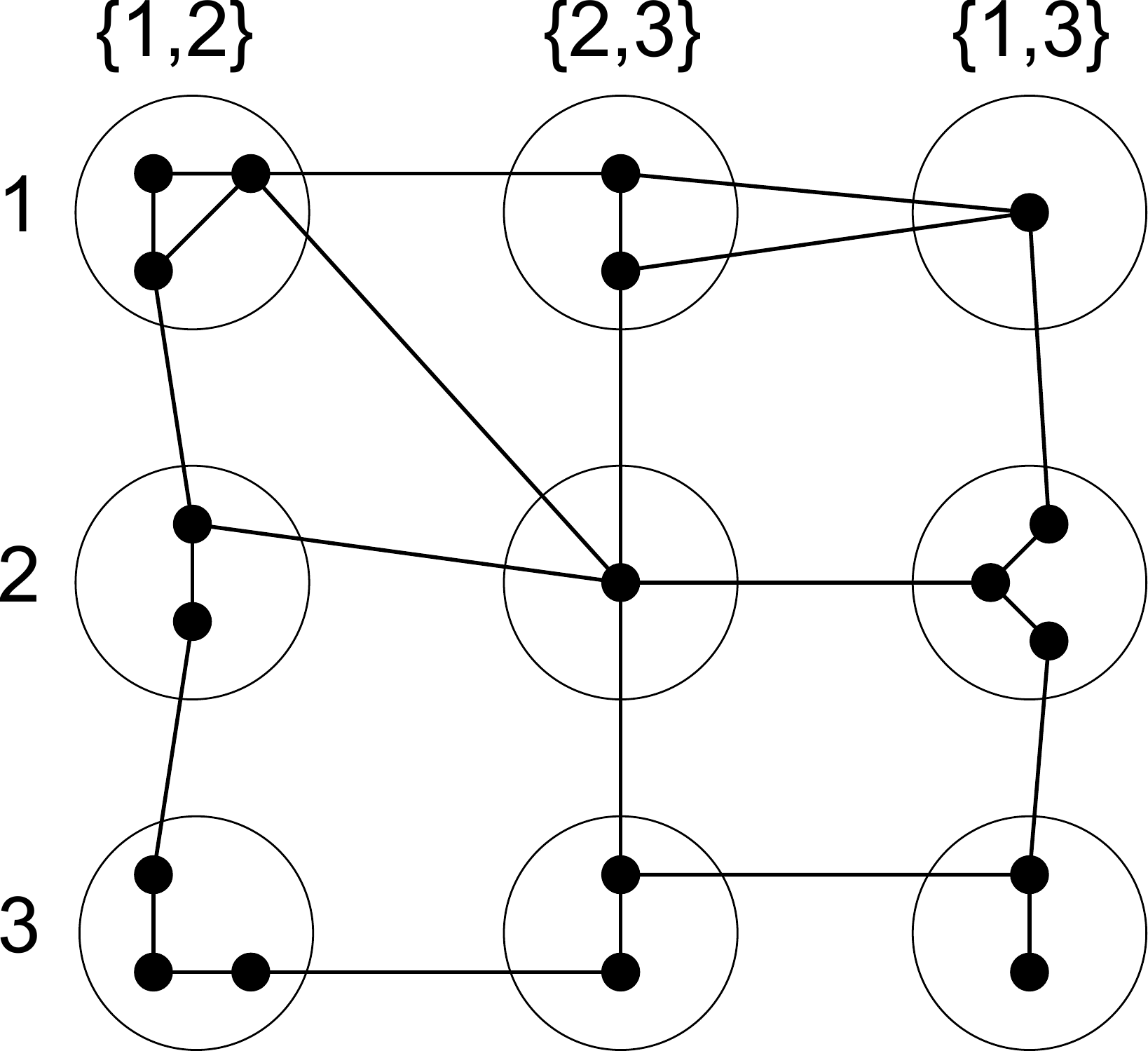}
\caption{An example pair of graphs $G$ (above) and $H$ (below), with $k=3$ (note that in general $H$ will not be based on a square grid, as $k \neq \binom{k}{2}$ for larger $k$).}
\label{constr-example}
\end{figure}

Let $A$ denote the $k \times \binom{k}{2}$ grid.  We will label each of the $k$ rows with a distinct element of the set $\{1,\ldots,k\}$, and each of the $\binom{k}{2}$ columns with a distinct unordered pair of elements from the same set.  We then denote by $a_{(i,\{j,l\})}$ the vertex in row $i$ and column $\{j,l\}$.

Since the grid $A$ is a minor of $H$, there exists a function $m: V(A) \rightarrow \mathcal{P}(V(H))$ such that the image of each vertex in $A$ induces a connected subgraph in $H$, distinct vertices of $A$ map to disjoint subsets of $V(H)$, and for each edge $aa' \in A$ there exist $u \in m(a)$ and $u' \in m(a')$ such that $uu' \in E(H)$.  Note that the union of the image of $m$ is not necessarily the whole of $V(H)$; if not, then we denote by $V_H'$ the vertices of $V(H) \setminus \bigcup_{a \in V(A)} m(a)$.

In our new graph $G_H$, we will have multiple copies of the subgraph $H[m(a)]$ of $H$ for each $a \in V(A)$, each one indexed by a pair consisting of a vertex and an edge from $G$. Specifically, for each $a_{(i,\{j,l\})} \in V(A)$, we will have a copy $H^{(v,e)}_{(i,\{j,l\})}$ of $H[m(a_{(i,\{j,l\})})]$ for each pair $(v,e) \in V(G) \times E(G)$ that satisfies the condition
\begin{equation}
\text{if } i \in \{j,l\} \text{ then } v \text{ is incident with } e.
\label{incidence-condition}
\end{equation}
In the example of Figure \ref{constr-example}, we will therefore have eight copies of $H[m(a_{(1,\{1,2\})})]$ (each of which is a triangle): $H^{(v_1,e_1)}_{(1,\{1,2\})}$, $H^{(v_1,e_3)}_{(1,\{1,2\})}$, $H^{(v_2,e_1)}_{(1,\{1,2\})}$, $H^{(v_2,e_2)}_{(1,\{1,2\})}$, $H^{(v_3,e_2)}_{(1,\{1,2\})}$, $H^{(v_3,e_3)}_{(1,\{1,2\})}$, $H^{(v_3,e_4)}_{(1,\{1,2\})}$, and $H^{(v_4,e_4)}_{(1,\{1,2\})}$.

For each $u \in m(a_{i,\{j,l\}})$, we denote by $u^{(v,e)}$ the corresponding vertex in $H^{(v,e)}_{(i,\{j,l\})}$ (note that each $u \in V(H)$ belongs to subgraphs $H_{(i,\{j,l\})}$ for only a single value of $(i,\{j,l\})$).  For each $(i,\{j,l\}) \in [k] \times [k]^{(2)}$, we denote by $V_{(i,\{j,l\})}$ the union of the following set:
\begin{align*}
\{V\left(H_{(i,\{j,l\})}^{(v,e)}\right): \quad & (v,e) \in V(G) \times E(G), \text{ and } \\ 
				& i \in \{j,l\} \implies v \text{ incident with } e\}.
\end{align*}

We now set 
$$V(G_H) = \bigcup_{(i,\{j,l\}) \in [k] \times [k]^{(2)}} V_{(i,\{j,l\})} \cup V_H'.$$
That is, the vertex-set of $G_H$ is made up of all the vertices in subgraphs $H_{(i,\{j,l\})}^{(v,e)}$, together with any vertices in $V_H'$.  

We now define the colouring $f$ of $V(G_H)$.  First, fix a colouring $\omega$ of $V(H)$ which gives every vertex a different colour from the set $\{1,\ldots,|V(H)|\}$.  For any $u \in V_H'$, we now set $f(u) = \omega(u)$, while for every vertex $u^{(v,e)} \in V(G_H)$ we set $f(u^{(v,e)}) = \omega(u)$.

To complete the definition of $G_H$, we now define its edge-set.  Let $u,w \in V(G_H)$.  Then $uw \in E(G_H)$ if and only if the following conditions are satisfied:
\begin{description}
\item[Condition 1:] $\omega^{-1}(f(u))\omega^{-1}(f(w)) \in E(H)$ (i.e.~we only join vertices if their colours occur at adjacent vertices in $H$), and
\item[Condition 2:] if $u \in H_{(i,\{j,l\})}^{(v,e)}$ and $w \in H_{(i',\{j',l'\})}^{(v',e')}$ then
\begin{enumerate}
\item if $i = i'$ then $v = v'$, and if $i < i'$ then $v \prec v'$, and 
\item if $\{j,l\} = \{j',l'\}$ then $e = e'$.
\end{enumerate}
\end{description}
Note that the second condition implies that there are no edges between $H_{(i,\{j,l\})}^{(v,e)}$ and $H_{(i,\{j,l\})}^{(v',e')}$ for $(v,e) \neq (v',e')$.  In Figures \ref{constr-eg1} and \ref{constr-eg2} we illustrate parts of this construction for the example in Figure \ref{constr-example} above, showing the edges between $V_{(1,\{1,2\})}$ and $V_{(2,\{1,2\})}$ (Figure \ref{constr-eg1}) and between $V_{(2,\{1,2\})}$ and $V_{(2,\{2,3\})}$ (Figure \ref{constr-eg2}); for clarity in the diagram, we label each subgraph $H_{(i,\{j,l\})}^{(v,e)}$ with $(v,e)$ only.

\begin{figure}[h]
\centering
\includegraphics[width = 0.7 \linewidth]{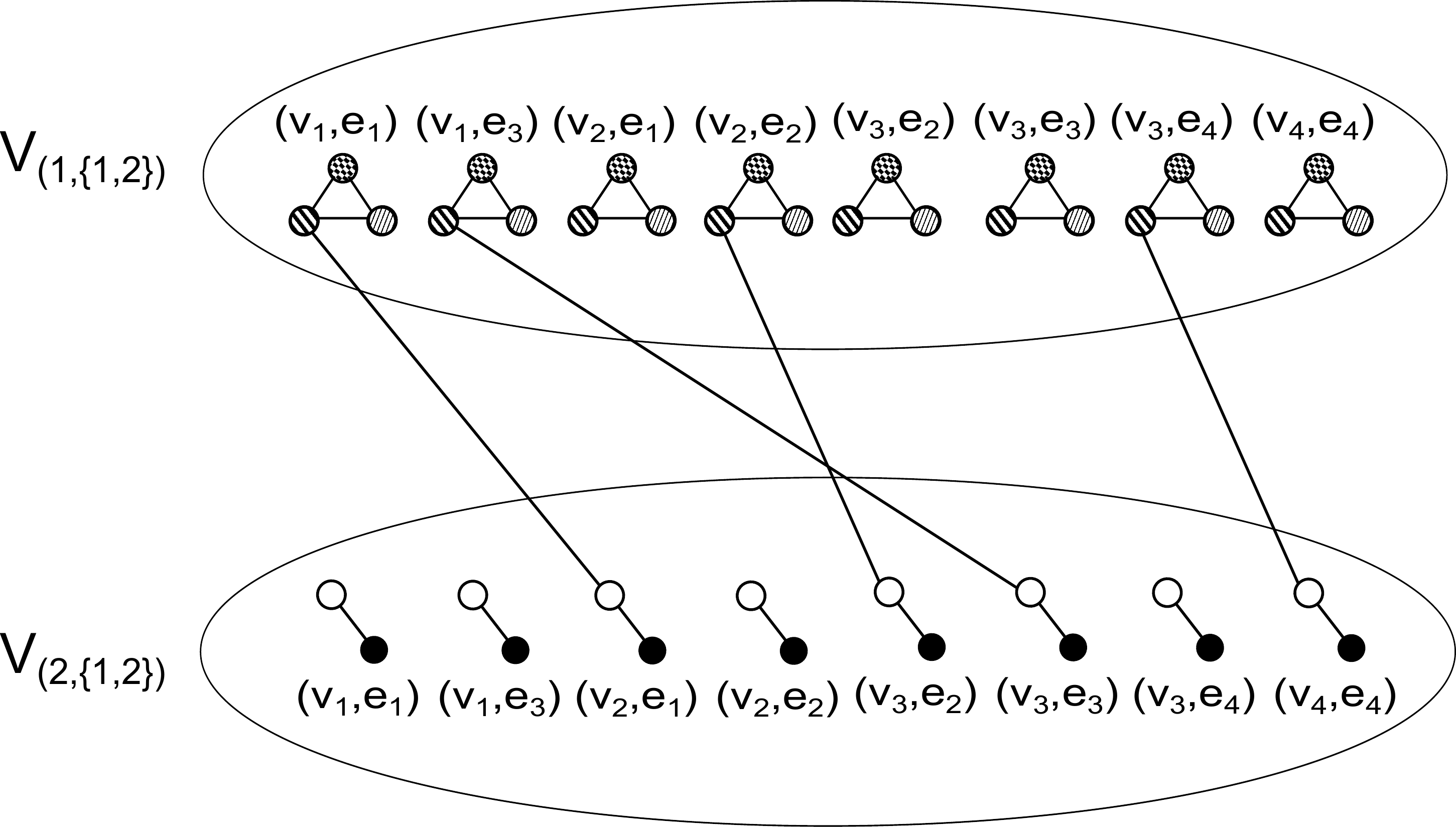}
\caption{Edges between $V_{(1,\{1,2\})}$ and $V_{(2,\{1,2\})}$}
\label{constr-eg1}
\end{figure}

\begin{figure}
\centering
\includegraphics[width = 0.5 \linewidth]{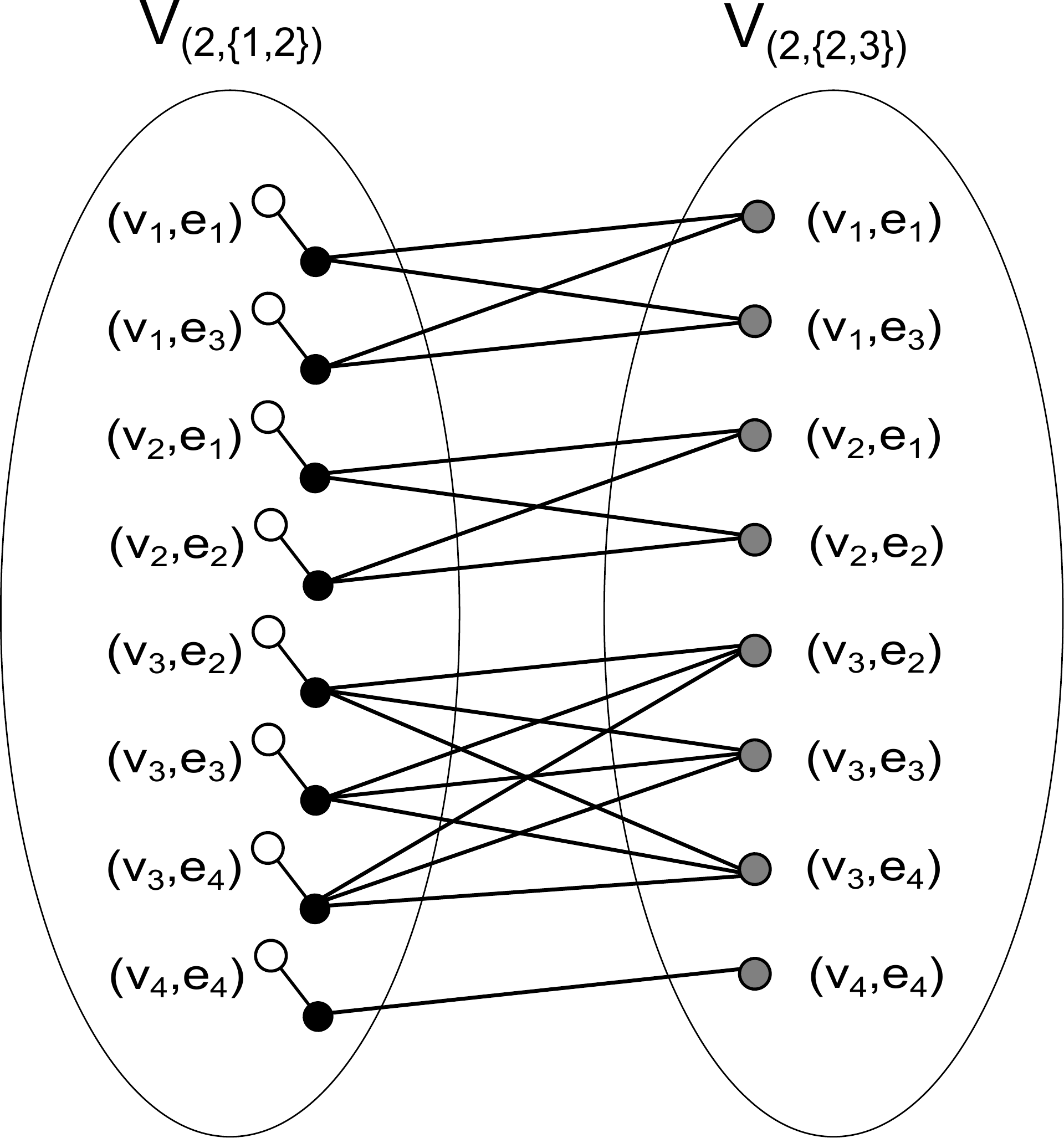}
\caption{Edges between $V_{(2,\{1,2\})}$ and $V_{(2,\{2,3\})}$}
\label{constr-eg2}
\end{figure}

Finally, observe that, given $G$ and $H$, the graph $G_H$ and its colouring $f$ can clearly be constructed in time $h(k) + n^{O(1)}$, where $h$ (a computable function of $k$) represents the time required to find a grid minor in $H$.

\subsubsection{Analysis of the construction}
\label{analysis}

In this section, we prove the important properties of the construction defined above; these properties will form the basis of the proofs of our complexity results in Section \ref{complexity-results}.  We begin by giving a necessary and sufficient condition for a colourful subgraph in $G_H$ to be isomorphic to $H$.

\begin{lma}
Let $Y$ be a colourful subset of $V(G_H)$.  Then $G_H[Y]$ is isomorphic to a subgraph of $H$.  Moreover, $G_H[Y]$ is isomorphic to $H$ if and only if, for every $u \in Y \cap H_{(i,\{j,l\})}^{(v,e)}$ and $w \in Y \cap H_{(i',\{j',l'\})}^{(v',e')}$ such that $\omega^{-1}(f(u))\omega^{-1}(f(w)) \in E(H)$, the following conditions are satisfied:
\begin{enumerate}
\item if $i = i'$ then $v = v'$, and if $i < i'$ then $v \prec v'$, and
\item if $\{j,l\} = \{j',l'\}$ then $e = e'$.
\end{enumerate}
\label{iso-condition}
\end{lma}
\begin{proof}
Let $\pi: Y \rightarrow V(H)$ be the mapping given by $\pi(u) = \omega^{-1}(f(u))$.  We claim first that $Y$ defines an isomorphism from $G_H[Y]$ to a subgraph of $H$.  To verify this claim, it suffices to check that, whenever $u,w \in Y$ with $uw \in E(G_H)$, we also have $\pi(u)\pi(w) \in E(H)$.  However, by the first condition in the definition of $E(G_H)$, we can only have $uw \in E(G_H)$ if $\omega^{-1}(f(u))\omega^{-1}(f(w)) \in E(H)$, or in other words if $\pi(u)\pi(w) \in E(H)$, as required.

Now, suppose that $Y$ is such that, for every $u \in Y \cap H_{(i,\{j,l\})}^{(v,e)}$ and $w \in Y \cap H_{(i',\{j',l'\})}^{(v',e')}$ with $\omega^{-1}(f(u))\omega^{-1}(f(w)) \in E(H)$, the following conditions are satisfied:
\begin{enumerate}
\item if $i = i'$ then $v = v'$, and if $i < i'$ then $v \prec v'$, and
\item if $\{j,l\} = \{j',l'\}$ then $e = e'$.
\end{enumerate}
We claim that in this case $\pi$ defines an isomorphism from $G_H[Y]$ to $H$.  This claim holds provided that, for all $u,w \in Y$ such that $\pi(u)\pi(w) \in E(H)$ we also have $uw \in E(G_H)$.  The assumption that $\pi(u)\pi(w) \in E(H)$ implies $uw$ satisfies Condition 1 in the definition of $E(G_H)$ above, and the condition that we are assuming is exactly Condition 2 in this definition; thus it follows immediately that in this situation we will have $uw \in E(G_H)$ as required.

Conversely, suppose that the colourful subset $Y \subset V(G_H)$ induces a subgraph isomorphic to $H$; we need to verify that, for every $u \in Y \cap H_{(i,\{j,l\})}^{(v,e)}$ and $w \in Y \cap H_{(i',\{j',l'\})}^{(v',e')}$ such that $\omega^{-1}(f(u))\omega^{-1}(f(w)) \in E(H)$, the conditions in the statement of the lemma are satisfied.  By definition of $\pi$ we know that, for any such $u$ and $w$, $\pi(u)\pi(w) \in E(H)$.  Moreover, since we have already established that $\pi$ defines an isomorphism from $G_H[Y]$ to a subgraph of $H$, it follows in this case that $\pi$ must be an isomorphism from $G_H[Y]$ to $H$; thus the fact that $\pi(u)\pi(w) \in E(H)$ implies that $uw \in E(G_H)$.  It then follows immediately from Condition 2 that 
\begin{enumerate}
\item if $i = i'$ then $v = v'$, and if $i < i'$ then $v \prec v'$, and
\item if $\{j,l\} = \{j',l'\}$ then $e = e'$,
\end{enumerate}
as required.
\end{proof}

Using this characterisation, we can now relate the number of colourful subgraphs in $G_H$ that are isomorphic to $H$ to the number of cliques in $G$.

\begin{lma}
Let $G$, $H$, $G_H$ and $f$ be as described above.  Then
$$\ColSubInd(H,G_H,f) = \Clique_k(G).$$
\label{count-cliques}
\end{lma}
\begin{proof}
We will prove this result by showing that there is a one-to-one correspondence between colourful copies of $H$ in $G_H$ with colouring $f$ and $k$-vertex subsets of $V(G)$ that induce cliques.  We begin by giving an injective map from the set of $k$-vertex subsets inducing cliques in $G$ to the colourful copies of $H$ in $G_H$ with colouring $f$; we will then proceed to give an injective map in the opposite direction.

Given a a subset $X \in V(G)^{(k)}$ that induces a clique in $G$, we denote by $x_1,\ldots,x_k$ the elements of $X$, with $x_1 \prec \ldots \prec x_k$.  We then claim that the mapping $\pi$, where
$$\pi(X) = V_H' \cup \bigcup_{\substack{i,j,l \in [k] \\ j \neq l}} V\left(H_{(i,\{j,l\})}^{(x_i, x_jx_l)}\right)$$
is an injective mapping from $k$-vertex subsets of $V(G)$ that induce cliques in $G$ to colourful copies of $H$ in $G_H$ (with respect to the colouring $f$).

Observe first that this mapping is well-defined: if $\{x_1,\ldots,x_k\}$ induces a clique in $G$, all edges $x_jx_l$ (with $j,l \in [k]$, $j \neq l$) are edges of $G$, and moreover $x_i$ is incident with $x_jx_l$ whenever $i \in \{j,l\}$, so the specified subgraphs $H_{(i,\{j,l\})}^{(x_i, x_jx_l)}$ do exist in $G_H$.  

It is immediate from the definition of $\pi$ that, for $X \neq X'$, we will have $\pi(X) \neq \pi(X')$.  Thus, to show that $\pi$ defines an injective map from the set of $k$-tuples inducing cliques in $G$ to the set of colourful copies of $H$ in $G_H$, it remains only to demonstrate that, whenever $X$ induces a clique in $G$, $\pi(X)$ must induce a colourful copy of $H$ in $G_H$ (with respect to $f$).

We verify first that $\pi(X)$ is indeed colourful with respect to $f$.  First note that all vertices of $V_H'$ have distinct colours, and do not share colours with any vertex in $V(G_H) \setminus V_H'$.  Next observe that, for each $(i,\{j,l\}) \in [k] \times [k]^{(2)}$, there exists a unique pair $(v,e) \in V(G) \times E(G)$ such that $\pi(X) \cap V_{(i,\{j,l\})} = V\left(H_{(i,\{j,l\})}^{(v,e)}\right)$.  Thus, $\pi(X)$ contains exactly one vertex $u^{(v,e)}$ (for some $(v,e) \in V(G) \times E(G)$) for each $u \in V(H) \setminus V_H'$, and so by definition of $f$ we see that $\pi(X)$ is indeed colourful with respect to $f$. 

It remains to verify that $\pi(X)$ induces a copy of $H$.  As we know that $\pi(X)$ is colourful, it suffices by Lemma \ref{iso-condition} to verify that, for every pair of vertices $u,w \in \pi(X) \setminus V_H'$ such that $u \in H_{(i,\{j,l\})}^{(v,e)}$, $w \in H_{(i',\{j',l'\})}^{(v',e')}$ and $\omega^{-1}(f(u))\omega^{-1}(f(w)) \in E(H)$, the following conditions are satisfied:
\begin{enumerate}
\item if $i = i'$ then $v = v'$, and if $i < i'$ then $v \prec v'$, and
\item if $\{j,l\} = \{j',l'\}$ then $e = e'$.
\end{enumerate}
However, these two conditions follow immediately from the definition of $\pi$, so $\pi(X)$ does indeed induce a colourful copy of $H$, as required.

We now proceed to define a mapping in the opposite direction, that is an injective mapping from colourful copies of $H$ in $G_H$ (with colouring $f_H$) to $k$-vertex subsets of $V(G)$ that induce cliques in $G$.  Suppose that the colourful subset $Y \subset V(G_H)$ induces a copy of $H$.  In order to define our mapping, we prove a series of claims.  We begin by arguing that the colourful subset $Y$ can be used to define a mapping from $[k] \times [k]^{(2)}$ to pairs in $V(G) \times E(G)$.

\begin{lblclaim}
For each $(i,\{j,l\}) \in [k] \times [k]^{(2)}$, there is a unique pair $(v,e) \in V(G) \times E(G)$ such that $Y \cap V_{(i,\{j,l\})} = V\left(H_{(i,\{j,l\})}^{(v,e)}\right)$.
\label{claim-unique-pair}
\end{lblclaim}
\begin{proof}[Proof of Claim \ref{claim-unique-pair}]
It follows from Lemma \ref{iso-condition} and the connectedness of $H[m(u)]$ for each $u \in V(H)$ that, for each $(i,\{j,l\}) \in [k] \times [k]^{(2)}$, there exists a unique pair $(v,e) \in V(G) \times E(G)$ such that $Y \cap V_{(i,\{j,l\})} \subseteq V\left(H_{(i,\{j,l\})}^{(v,e)}\right)$.  Thus, as $Y$ is colourful, we must in fact have $Y \cap V_{(i,\{j,l\})} = V\left(H_{(i,\{j,l\})}^{(v,e)}\right)$ for this $(v,e)$, as required (otherwise we would miss some colour). 
\renewcommand{\qedsymbol}{$\square$ (Claim \ref{claim-unique-pair})}
\end{proof}

Hence we can define a mapping $\sigma_Y: [k] \times [k]^{(2)} \rightarrow V(G) \times E(G)$ by setting $\sigma_Y(i,\{j,l\})$ to be this unique pair $(v,e)$ such that $Y \cap V_{(i,\{j,l\})} = V\left(H_{(i,\{j,l\})}^{(v,e)}\right)$; we denote by $\sigma_Y^1(i,\{j,l\})$ and $\sigma_Y^2(i,\{j,l\})$ respectively the first and second elements of this pair (so $\sigma_Y^1(i,\{j,l\}) \in V(G)$, and $\sigma_Y^2(i,\{j,l\}) \in E(G)$).  

We now show that the value of $\sigma_Y^1$ depends only on its first argument, and $\sigma_Y^2$ only on its second.

\begin{lblclaim}
For each $i \in [k]$ and every $\{j,l\},\{j',l'\} \in [k]^{(2)}$, we have $\sigma_Y^1((i,\{j,l\})) = \sigma_Y^1((i,\{j',l'\}))$.  Similarly, for each $\{j,l\} \in [k]^{(2)}$ and every $i,i' \in [k]$, we have $\sigma_Y^2((i,\{j,l\})) = \sigma_Y^2((i',\{j,l\}))$.
\label{claim-sigma-args}
\end{lblclaim}
\begin{proof}[Proof of Claim \ref{claim-sigma-args}]
To prove this claim we suppose, for a contradiction, there exist $\{j,l\}$ and $\{j',l'\}$ such that $\sigma_Y^1((i,\{j,l\})) \neq \sigma_Y^1((i,\{j',l'\}))$; we may assume without loss of generality that $a_{(i,\{j,l\})}$ and $a_{(i,\{j',l'\})}$ are adjacent in the grid $A$.  Now, by definition of the minor map $m$, there exist vertices $z,z' \in V(H)$ such that $z \in m(a_{(i,\{j,l\})})$, $z' \in m(a_{(i,\{j',l'\})})$ and $zz' \in E(H)$.  Let $y$ and $y'$ be the vertices of $Y$ whose colours match $z$ and $z'$ respectively, that is $y = f|_Y^{-1}\omega(z)$ and $y' = f|_Y^{-1}\omega(z')$ (as $Y$ is colourful, the restriction of $f$ to $Y$ is a bijection), so $\omega^{-1}(f(y)) \omega^{-1}(f(y')) \in E(H)$.  Note that we must have $y \in H_{(i,\{j,l\})}^{\sigma_Y((i,\{j,l\}))}$ and $y' \in H_{(i,\{j',l'\})}^{\sigma_Y((i,\{j',l'\}))}$.  By our initial assumption, $\sigma_Y^1((i,\{j,l\})) \neq \sigma_Y^1((i,\{j',l'\}))$, but by Lemma \ref{iso-condition} this then means that $Y$ cannot induce a copy of $H$ in $G_H$, giving a contradiction.  Thus we know that, for each $i \in [k]$, we must have $\sigma_Y^1((i,\{j,l\})) = \sigma_Y^1((i,\{j',l'\}))$ for all $\{j,l\},\{j',l'\} \in [k]^{(2)}$.  An analogous argument can be used to show that, for each $\{j,l\} \in [k]^{(2)}$, we must have $\sigma_Y^2((i,\{j,l\})) = \sigma_Y^2((i',\{j,l\}))$ for every $i,i' \in [k]$.
\renewcommand{\qedsymbol}{$\square$ (Claim \ref{claim-sigma-args})}
\end{proof}

Hence we can define the mapping $\tau_Y^1: [k] \rightarrow V(G)$ by setting $\tau_Y^1(i)$ to be the unique vertex $v \in V(G)$ such that, for any value of $\{j,l\} \in [k]^{(2)}$, $\sigma_Y^1((i,\{j,l\})) = v$, and the mapping $\tau_Y^2: [k]^{(2)} \rightarrow E(G)$ by setting $\tau_Y^2(\{j,l\})$ to be the unique edge $e \in E(G)$ such that, for any value of $i \in [k]$, $\sigma_Y^2((i,\{j,l\})) = e$.

Next, we demonstrate that the mapping $\tau_Y^1$ is injective.

\begin{lblclaim}
For $i \neq i'$, we have $\tau_Y^1(i) \neq \tau_Y^1(i')$.
\label{claim-tau-inj}
\end{lblclaim}
\begin{proof}[Proof of Claim \ref{claim-tau-inj}]
We argue that, for $1 \leq i \leq k-1$, we must have $\tau_Y^1(i) \prec \tau_Y^1(i+1)$, implying that $\tau_Y^1(1) \prec \tau_Y^1(2) \prec \cdots \prec \tau_Y^1(k)$ and in particular that no two of these vertices are the same.  As in the proof of Claim \ref{claim-sigma-args}, we observe that by definition of the minor map $m$, for each $i \in \{1,\ldots,k-1\}$, there exist $z,z' \in V(H)$ such that $z \in m(a_{(i,\{1,2\})})$, $z' \in m(a_{(i+1,\{1,2\})})$ and $zz' \in E(H)$.  As before, we set $y = f|_Y^{-1}\omega(z)$ and $y' = f|_Y^{-1}\omega(z')$, and observe that $\omega^{-1}(f(y))\omega^{-1}(f(y')) \in E(H)$.   Note also, by definition of $\tau_Y^1$ and $\tau_Y^2$, that we have $y \in H_{(i,\{1,2\})}^{(\tau_Y^1(i),\tau_Y^2(\{1,2\}))}$ and $y' \in H_{(i+1,\{1,2\})}^{(\tau_Y^1(i+1),\tau_Y^2(\{1,2\}))}$.  Since $i < i+1$, it now follows from Lemma \ref{iso-condition} that $\tau_Y^1(i) \prec \tau_Y^1(i+1)$, as required.
\renewcommand{\qedsymbol}{$\square$ (Claim \ref{claim-tau-inj})}
\end{proof}

We now demonstrate a relationship between the mappings $\tau_Y^1$ and $\tau_Y^2$.

\begin{lblclaim}
Let $j,l \in [k]$ with $j \neq l$.  Then $\tau_Y^2(\{j,l\}) = \tau_Y^1(j)\tau_Y^2(l)$.
\label{claim-tau1,2}
\end{lblclaim}
\begin{proof}[Proof of Claim \ref{claim-tau1,2}]
Set $e = \tau_Y^2(\{j,l\})$.  Then, by definition, $\sigma_Y(j,\{j,l\}) = (\tau_Y^1(j),e)$ and $\sigma_Y(l,\{j,l\}) = (\tau_Y^1(l),e)$.  Thus, both subgraphs $H_{(j,\{j,l\})}^{(\tau_Y^1(j),e)}$ and $H_{(l,\{j,l\})}^{(\tau_Y^1(l),e)}$ must exist in $G_H$; by definition of $V(G_H)$, since $j,l \in \{j,l\}$, this implies that both $\tau_Y^1(j)$ and $\tau_Y^1(l)$ must be incident with $e$.  By Claim \ref{claim-tau-inj}, we know that $\tau_Y^1(j) \neq \tau_Y^1(l)$, so this is only possible if in fact $e = \tau_Y^1(j) \tau_Y^1(l)$, as required.
\renewcommand{\qedsymbol}{$\square$ (Claim \ref{claim-tau1,2})}
\end{proof}

Using the facts proved so far, we can now give a mapping from colourful subsets $Y \subset V(G_H)$ with $G_H[Y] \cong H$ to $k$-tuples of $V(G)$ that induce cliques.

\begin{lblclaim}
$V_Y = \{\tau_Y^1(i): 1 \leq i \leq k\}$ induces a $k$-clique in $G$.
\label{claim-clique}
\end{lblclaim}
\begin{proof}[Proof of Claim \ref{claim-clique}]
We already know from Claim \ref{claim-tau-inj} that, for $i \neq i'$, $\tau_Y^1(i) \neq \tau_Y^1(i')$ and hence that $|V_Y| = k$; thus it suffices to verify that, for each $1 \leq i < i' \leq k$, $\tau_Y^1(i)\tau_Y^1(i') \in E(G_H)$.  But we know from Claim \ref{claim-tau1,2} that $\tau_Y^1(i)\tau_Y^1(i') = \tau_Y^2(\{i,i'\}) \in E(G_H)$, completing the proof of the claim.
\renewcommand{\qedsymbol}{$\square$ (Claim \ref{claim-clique})}
\end{proof}

We have now given a well-defined mapping $Y \mapsto (\tau_Y^1(1),\ldots,\tau_Y^1(k))$ from colourful subsets $Y$ of $V(G_H)$ with $G_H[Y] \cong H$ to $k$-tuples of $V(G)$ that induce cliques, so in order to complete the proof of the lemma it remains only to argue that this mapping is injective.

\begin{lblclaim}
The mapping $Y \mapsto (\tau_Y^1(1),\ldots,\tau_Y^1(k))$ is injective.
\label{claim-map-inj}
\end{lblclaim}
\begin{proof}[Proof of Claim \ref{claim-map-inj}]
Suppose that $Y$ and $Y'$ are distinct colourful subsets of $G_H$ with respect to the colouring $f$, such that $G_H[Y],G_H[Y'] \cong H$.  By the reasoning above, we know that there is a function $\tau_{Y}^1: [k] \rightarrow V(G)$ such that, for each $(i,\{j,l\}) \in [k] \times [k]^{(2)}$, 
$$Y \cap V_{(i,\{j,l\})} = V\left(H_{(i,\{j,l\})}^{(\tau_{Y}^1(i),\tau_{Y}^1(j)\tau_{Y}^1(l))}\right)$$
and 
$$Y' \cap V_{(i,\{j,l\})} = V\left(H_{(i,\{j,l\})}^{(\tau_{Y'}^1(i),\tau_{Y'}^1(j)\tau_{Y'}^1(l))}\right).$$  
Since $Y$ and $Y'$ are both colourful, we know that
\begin{align*}
Y & = V_H' \cup \bigcup_{(i,\{j,l\}) \in [k] \times [k]^{(2)}} \left(Y \cap V_{(i,\{j,l\})}\right) \\
  & = V_H' \cup \bigcup_{(i,\{j,l\}) \in [k] \times [k]^{(2)}} V\left(H_{(i,\{j,l\})}^{(\tau_{Y}^1(i),\tau_{Y}^1(j)\tau_{Y}^1(l))}\right),
\end{align*}
and similarly
\begin{align*}
Y' & = V_H' \cup \bigcup_{(i,\{j,l\}) \in [k] \times [k]^{(2)}} \left(Y' \cap V_{(i,\{j,l\})}\right) \\
   & = V_H' \cup \bigcup_{(i,\{j,l\}) \in [k] \times [k]^{(2)}} V\left(H_{(i,\{j,l\})}^{(\tau_{Y'}^1(i),\tau_{Y'}^1(j)\tau_{Y'}^1(l))}\right).
\end{align*}
Thus, if $\tau_Y^1 = \tau_{Y'}^1$, we would have $Y=Y'$, a contradiction to our initial assumption.  Hence $\tau_Y^1 \neq \tau_{Y'}^1$, implying immediately that $(\tau_Y(1),\ldots,\tau_Y^1(k)) \neq (\tau_{Y'}(1),\ldots,\tau_{Y'}^1(k))$, as required.
\renewcommand{\qedsymbol}{$\square$ (Claim \ref{claim-map-inj})}
\end{proof}
This claim completes the proof that there is a one-to-one correspondence between colourful subsets $Y \subset V(G_H)$ with $G_H[Y] \cong H$ and $k$-tuples of vertices that induce cliques in $G$, and hence completes the proof that
$$\ColSubInd(H,G_H,f) = \Clique_k(G),$$
as required.
\end{proof}

Before proving our complexity results in the next section, we give a simple fact about the number of colourful copies of graphs from some collection $\mathcal{H}$ in $G_H$, provided that $\mathcal{H}$ satisfies certain conditions.

\begin{lma}
Let $\mathcal{H}$ be a collection of labelled graphs on $k$ vertices such that $(H,\pi) \in \mathcal{H}$ and there is no proper subgraph $H'$ of $H$ such that, for some labelling $\pi'$, $(H',\pi') \in \mathcal{H}$.  Then
$$\ColStrEmb(\mathcal{H},G_H,f) = \ColStrEmb(\mathcal{H}^H,G_H,f).$$
\label{H-minimal}
\end{lma}
\begin{proof}
We know from Lemma \ref{iso-condition} that any colourful subset in $G_H$ induces a subgraph of $H$; thus, by our assumption that $\mathcal{H}$ does not contain any proper subgraphs of $H$, it follows immediately that the number of colourful labelled copies of graphs from $\mathcal{H}$ in $G_H$ (with respect to the colouring $f$) is exactly equal to the number of colourful labelled copies of graphs from $\mathcal{H}^H$ in $G_H$ (with respect to the colouring $f$).
\end{proof}

\subsection{Complexity results}
\label{complexity-results}

In this section, we use the construction defined in Section \ref{construction}, and in particular its properties established in Section \ref{analysis}, to prove our new complexity results.  In several of the proofs below we will also need a result from \cite{bddlayers}:

\begin{lma}
Let $\mathcal{H}$ be a collection of labelled graphs, and $(H,\pi) \in \mathcal{H}$ a labelled $k$-vertex graph.  Set 
\begin{align*}
\alpha_H = |\{\sigma : \quad & \sigma \text{ a permutation on $[k]$, $\exists (H,\pi') \in \mathcal{H}^H$ such that} \\
				    		 &  \text{$\pi' \circ \sigma^{-1} \circ \pi^{-1}$ defines an automorphism on $H$}\}|.
\end{align*}
Then, for any graph $G$, 
$$\StrEmb(\mathcal{H}^H,G) = \alpha_H \cdot \SubInd(H,G),$$
Moreover, if $G$ is equipped with a $k$-colouring $f$, then
$$\ColStrEmb(\mathcal{H}^H,G,f) = \alpha_H \cdot \ColSubInd(H,G,f).$$
\label{emb->subg}
\end{lma}

Note that the value of $\alpha_H$, as defined in the statement of this lemma, can be computed from $H$ and $\mathcal{H}$ in time bounded only by some computable function of $k$.

\subsubsection*{Exact counting}

We begin with the hardness result for exact counting when $\min^*(\Phi)$ does not have bounded treewidth; this result applies in both the multicolour and uncoloured settings.  Note that we do not need to assume that $\Phi$ is monotone.

\begin{thm}
Let $\Phi$ be a family $(\phi_1,\phi_2,\ldots)$ of functions $\phi_k: \mathcal{L}(k) \rightarrow \{0,1\}$, such that the function mapping $k \mapsto \phi_k$ is computable.  Suppose that $\min^*(\Phi)$ has unbounded treewidth.  Then \paramcount{\genprob}$(\Phi)$ and \paramcount{\genprobcol}$(\Phi)$ are \#W[1]-complete.
\label{exact-hard}
\end{thm}
\begin{proof}
We prove the result by means of an fpt-Turing reduction from the \#W[1]-complete problem \paramcount{Clique}.  Note that, by Lemma \ref{uncol-col}, it suffices to give a reduction from \paramcount{Clique} to \paramcount{\genprobcol}$(\Phi)$.

Let the graph $G$ be the input to an instance of \paramcount{Clique} with parameter $k$.  By Proposition \ref{contains-grid}, we know that there is $k'$ (bounded by some computable function of $k$) such that some $H \in \min^*(\phi_k)$ contains the $(k \times \binom{k}{2})$ grid as a minor and, moreover, that there is no proper subgraph $H'$ of $H$ which satisfies $\phi_{k'}$ (with any labelling).  Let $\mathcal{H}$ be the set of all $k'$-vertex labelled graphs satisfying $\phi_{k'}$.

Recall that the graph $G_H$ and colouring $f$ can be constructed in time $h(k) + n^{O(1)}$, where $h$ is a computable function.  Given the graph $G_H$, the colouring $f$ and an oracle to \paramcount{\genprobcol}$(\Phi)$, it is now straightforward to compute the number of $k$-cliques in $G$, since
\begin{align*}
\Clique_k(G,f) & = \ColSubInd(H,G_H,f) \\
				& \qquad \qquad \qquad \qquad \mbox{by Lemma \ref{count-cliques}} \\
				& = \frac{1}{\alpha_H} \ColStrEmb(\mathcal{H}^H,G_H,f) \\
				& \qquad \qquad \qquad \qquad \mbox{by Lemma \ref{emb->subg}} \\
				& = \frac{1}{\alpha_H} \ColStrEmb(\mathcal{H},G_H,f) \\
				& \qquad \qquad \qquad \qquad \mbox{by Lemma \ref{H-minimal}},
\end{align*}
where
\begin{align*}
\alpha_H = |\{\sigma : \quad & \sigma \text{ a permutation on $[k]$, $\exists (H,\pi') \in \mathcal{H}^H$ such that} \\
				    		 &  \text{$\pi' \circ \sigma^{-1} \circ \pi^{-1}$ defines an automorphism on $H$}\}|.
\end{align*}
Recall that $\alpha_H$ can be computed from $H$ and $\mathcal{H}$ in time depending only on a computable function of $k'$.  Note also that the value of the parameter in the oracle call used to evaluate $\ColStrEmb(\mathcal{H},G_H,f)$ is at most $g(w(\binom{k}{2}))$, which is a computable function of $k$ only.  This therefore gives an fpt-Turing reduction from \paramcount{Clique} to \paramcount{\genprobcol}$(\Phi)$, as required.
\end{proof}

\subsubsection*{Decision and approximate counting in the multicolour setting}

We now turn our attention to the decision and approximate counting questions.  We begin by giving a sufficient condition for \paramdec{\genprobcol}$(\Phi)$ to be W[1]-complete, which does not make any assumptions about the monotonicity of $\Phi$.  For the definition of fpt-reductions for parameterised decision problems, we again refer the reader to \cite{flumgrohe}.

\begin{thm}
Let $\Phi$ be a family $(\phi_1,\phi_2,\ldots)$ of functions $\phi_k: \mathcal{L}(k) \rightarrow \{0,1\}$, such that the function mapping $k \mapsto \phi_k$ is computable, and suppose that $\min^*(\Phi)$ has unbounded treewidth.  Then \paramdec{\genprobcol}$(\Phi)$ is W[1]-complete.
\label{decision-hard}
\end{thm}
\begin{proof}
We prove W[1]-completeness by means of an fpt-reduction from \paramdec{Clique}.  Let $G$  be the input in an instance of \paramdec{Clique} with parameter $k$.  By assumption, $\min^*(\Phi)$ has unbounded treewidth, so by Proposition \ref{contains-grid} there exists $k' \in \mathbb{N}$ (where $k'$ is at most some computable function of $k$) such that some $H \in \min^*(\phi_{k'})$ contains the $(k \times \binom{k}{2})$ grid as a minor.

Now, we can construct the graph $G_H$ and colouring $f$ in time $h(k) + n^{O(1)}$, where $h$ is a computable function.  Moreover, we know from Lemma \ref{count-cliques} that 
$$ \ColSubInd(H,G_H,f) = \Clique_k(G),$$
so in particular the number of colourful subgraphs of $G_H$ (with respect to the colouring $f$) that are isomorphic to $H$ is non-zero if and only if the number of $k$-cliques in $G$ is non-zero.  Since the parameter, $k'$, in the instance of \paramdec{\genprobcol}$(\Phi)$ is bounded by a computable function of $k$, this gives a fpt-reduction from \paramdec{Clique} to \paramdec{\genprobcol}$(\Phi)$, as required, completing the proof of the theorem.
\end{proof}

Together with Corollary \ref{bddtw-mondec}, this immediately implies a dichotomy for the special case in which $\Phi$ is a uniformly monotone property (recall from the definition of a uniformly monotone property $\Phi$ that the class of labelled graphs $\min(\Phi)$ has bounded treewidth if and only if the class of unlabelled graphs $\min^*(\Phi)$ has bounded treewidth).

\begin{cor}
Let $\Phi$ be a family $(\phi_1,\phi_2,\ldots)$ of functions $\phi_k: \{0,1\}^{\binom{k}{2}} \rightarrow \{0,1\}$, such that the function mapping $k \mapsto \phi_k$ is computable, and suppose further that $\Phi$ is a uniformly monotone property.  Then \paramdec{\genprobcol}$(\Phi)$ is in FPT if there exists an integer $t$ such that all graphs in $\min(\Phi)$ have treewidth at most $t$; otherwise \paramdec{\genprobcol}$(\Phi)$ is W[1]-complete.
\label{decision-dichot}
\end{cor}

Subject to the assumption that W[1] is not equal to FPT under randomised parameterised reductions, Theorem \ref{decision-hard} also implies an inapproximability result for the counting problem in the multicolour setting.

\begin{cor}
Assume that W[1] is not equal to FPT under randomised parameterised reductions, and let $\Phi$ be a family $(\phi_1,\phi_2,\ldots)$ of functions $\phi_k: \mathcal{L}(k) \rightarrow \{0,1\}$ such that the function $k \mapsto \phi_k$ is computable.  Then, if $\min^*(\Phi)$ has unbounded treewidth, there is no FPTRAS for \paramcount{\genprobcol}$(\Phi)$.
\label{approx-hard}
\end{cor}
\begin{proof}
This follows immediately from Theorem \ref{decision-dichot} above, together with Lemma \ref{no-decision}.
\end{proof}

In the case that $\Phi$ is a uniformly monotone property, we also obtain a dichotomy result for the approximability of the counting problem.

\begin{cor}
Assume that W[1] is not equal to FPT under randomised parameterised reductions.  Let $\Phi$ be a family $(\phi_1,\phi_2,\ldots)$ of functions $\phi_k: \mathcal{L}(k) \rightarrow \{0,1\}$ such that the function $k \mapsto \phi_k$ is computable, and suppose further that $\Phi$ is a uniformly monotone property.  Then there is an FPTRAS for \paramcount{\genprobcol}$(\Phi)$ if and only if $\min(\Phi)$ has bounded treewidth.
\end{cor}
\begin{proof}
The existence of an FPTRAS is precisely Theorem \ref{tw-fptras}, and the reverse direction is Theorem \ref{approx-hard} above.
\end{proof}

\section{Future directions}
\label{future}

There has been significant recent progress in the area of parameterised subgraph counting, but there remain many fundamental 
open questions, a number of which have already been touched on above.  In this section we concentrate on uniformly monotone properties, which formed the main focus of this paper, and identify the key open questions.

We begin by considering the situation for exact counting.  The only problems known to admit fpt-algorithms for exact counting are the ``trivial'' problems where the property can be decided, for any given $k$-vertex labelled subgraph, by examining edges incident with only a constant-sized subset of the vertices; in the case of uniformly monotone properties, this corresponds precisely to those properties $\Phi$ such that $\min(\Phi)$ is a class of graphs with bounded vertex-cover number.  It is straightforward to adapt existing algorithms (such as the one given in \cite{radu14} to solve \paramcount{Sub}$(\mathcal{H})$ when $\mathcal{H}$ is a class of graphs of bounded vertex-cover number) to give polynomial-time algorithms for such trivial properties.

However, the best general hardness result for uniformly monotone problems is Theorem \ref{exact-hard}, which shows that if $\min(\Phi)$ is a class of graphs of unbounded treewidth then \paramcount{\genprob}$(\Phi)$ is \#W[1]-complete.  It is known that this hardness result is not the best possible: \paramcount{Connected Induced Subgraph} is an example of a uniformly monotone property such that all minimal elements have bounded treewidth, but which is nevertheless known to be \#W[1]-complete \cite{connected}.

The recent result of Curticapean and Marx \cite{radu14} is a significant breakthrough towards closing this gap, giving a complexity dichotomy for the problem \paramcount{Sub}$(\mathcal{C})$ (showing in fact that, unless FPT $\neq$ \#W[1], the trivial problems are the only ones to admit fpt-algorithms in this setting).  However, the techniques used seem to be very specific to the situation in which there is a unique minimal element satisfying the property $\Phi$, and so do not transfer easily to the more general case of uniformly monotone properties.  Nevertheless, we conjecture that the same dichotomy holds for uniformly monotone properties:

\begin{conj}
Suppose that $\Phi$ is a uniformly monotone property, and that there is no constant $c$ such that all graphs in $\min(\Phi)$ have vertex-cover number at most $c$.  Then \paramcount{\genprob}$(\Phi)$ is \#W[1]-complete.
\end{conj}

Considering now the situation for decision problems, we were able to give a dichotomy result for the multicolour version (Theorem \ref{decision-dichot}), showing that \paramdec{\genprobcol}$(\Phi)$ is W[1]-complete if $\min(\Phi)$ is a class of graphs having unbounded treewidth, and belongs to FPT otherwise.  However, there remains a gap when we consider the uncoloured version: we know (Corollary \ref{bddtw-mondec}) that \paramdec{\genprob}$(\Phi)$ is in FPT if $\min(\Phi)$ has bounded treewidth, but the corresponding hardness result does not extend to the uncoloured version.  Flum and Grohe \cite{flumgrohe} conjectured that \paramdec{Sub}$(\mathcal{C})$ is W[1]-complete whenever $\mathcal{C}$ is a class of graphs having unbounded treewidth, mentioning in particular the case in which $\mathcal{C} = \{K_{k,k}: k \in \mathbb{N}\}$ (the set of all complete bipartite graphs with $k$ vertices in each vertex-class); this specific problem was recently shown to be W[1]-complete by Lin \cite{lin14}, but the general conjecture remains open.  We extend this conjecture to encompass all uniformly monotone properties:

\begin{conj}
Suppose that $\Phi$ is a uniformly monotone property, and that there is no constant $c$ such that all graphs in $\min(\Phi)$ have treewidth at most $c$.  Then \paramdec{\genprob}$(\Phi)$ is W[1]-complete.
\end{conj}

Observe that, by Lemma \ref{no-decision}, and subject to the assumption that W[1] is not equal to FPT under randomised parameterised reductions, this would also imply a dichotomy for the existence of an FPTRAS for uniformly monotone properties.  

It was noted in Section \ref{relationships} that recent developments in the theory of parameterised random complexity classes \cite{montoya13,chauhan14} are based on a more restricted notion of parameterised randomisation than is allowed in the definition of an FPTRAS.  An interesting general direction for future research in parameterised counting complexity, therefore, would be to develop an alternative notion of approximability for parameterised counting problems that places corresponding restrictions on the permitted use of randomness;  this would allow more precise conclusions to be drawn about the relationship between the approximability of counting problems and the complexity of the corresponding decision problems.  Having developed such a framework for approximability, a natural first question to consider would be whether the parameterised subgraph counting problems admitting an FPTRAS also satisfy this stronger notion of approximability.

\section*{Acknowledgements}

I am very grateful to Mark Jerrum for his helpful comments on several previous drafts of this paper.

\end{document}